\documentclass[11pt]{article}
\usepackage{amsmath,amssymb,amsfonts,amsthm,epsfig,bm,xspace,epsfig,latexsym,hyperref}
\newtheorem{theorem}{Theorem}[section]
\newtheorem{lemma}[theorem]{Lemma}

\newtheorem{proposition}[theorem]{Proposition}
\newtheorem{corollary}[theorem]{Corollary}
\newtheorem{fact}[theorem]{Fact}
\newtheorem{definition}[theorem]{Definition}
\newtheorem{remark}[theorem]{Remark}

\newcommand{\ignore}[1]{}
\newcommand{\rnote}[1]{}
\newcommand{\dnote}[1]{}
\newcommand{\pnote}[1]{}
\newcommand{\todobox}[1]{}

\renewcommand{\Pr}{\mathop{\bf Pr\/}}
\newcommand{\E}{\mathop{\bf E\/}}

\newcommand{\Cov}{\mathop{\bf Cov\/}}
\newcommand{\cov}{\mathop{\bf Cov\/}}

\newcommand{\R}{\mathbb R}
\newcommand{\N}{\mathbb N}
\newcommand{\F}{\mathbb F}

\newcommand{\eps}{\epsilon}

\newcommand{\sgn}{\mathrm{sgn}}
\newcommand{\vsgn}{\overrightarrow{\mathrm{sgn}}}
\newcommand{\poly}{\mathrm{poly}}
\newcommand{\polylog}{\mathrm{polylog}}

\newcommand{\littlesum}{\mathop{{\textstyle \sum}}}

\newcommand{\wt}{\widetilde}

\newcommand{\calL}{{\cal L}}
\newcommand{\calF}{{\cal F}}

\newcommand{\calH}{{\cal H}}
\newcommand{\calN}{{\cal N}}

\newcommand{\calP}{{\cal P}}

\newcommand{\calD}{{\cal D}}
\newcommand{\calC}{{\cal C}}

\newcommand{\bh}{\boldsymbol{h}}
\newcommand{\bp}{\boldsymbol{p}}

\newcommand{\bx}{\boldsymbol{x}}
\newcommand{\bw}{\boldsymbol{w}}
\newcommand{\by}{\boldsymbol{y}}
\newcommand{\bz}{\boldsymbol{z}}

\newcommand{\bX}{\boldsymbol{X}}
\newcommand{\bY}{\boldsymbol{Y}}
\newcommand{\bZ}{\boldsymbol{Z}}

\newcommand{\bH}{\boldsymbol{H}}
\newcommand{\bS}{\boldsymbol{S}}
\newcommand{\bT}{\boldsymbol{T}}
\newcommand{\bG}{\boldsymbol{G}}
\newcommand{\abs}[1]{\left\lvert #1 \right\rvert}           
\newcommand{\len}[1]{\left\lVert #1 \right\rVert}           
           
\newcommand{\rnorm}[1]{\pmb{\lVert} #1 \pmb{\rVert}} 
\newcommand{\nice}{b}
\newcommand{\bThm}{\boldsymbol{X}}
\newcommand{\bDerand}{\boldsymbol{Y}}

\newcommand{\Ind}[1]{{\bf 1}\!\left[#1\right]}       
\newcommand{\sign}{\mathrm{sign}}

\newcommand{\zo}{\{0,1\}}
\newcommand{\set}[1]{\{ #1 \}}
\newcommand{\pmo}{\ensuremath \{ -1, 1\}}

\newcommand{\domain}{\Omega}

\newcommand{\aaa}{a}
\newcommand{\bbb}{b}
\newcommand{\btheta}{\boldsymbol{\theta}}
\newcommand{\littleprod}{\mathop{{\textstyle \prod}}}

\newcounter{this-list}
\newenvironment{tightenumerate}{
\vspace{2pt}
\begin{list}{\arabic{this-list}.}{\usecounter{this-list}
                                 \setcounter{this-list}{0}
  \setlength{\itemsep}{0pt}%
  \setlength{\parsep}{0pt}%
  \setlength{\topsep}{0pt}%
    \setlength{\partopsep}{0pt}%
  \setlength{\leftmargin}{3.5ex}%
  \setlength{\labelwidth}{4.5ex}%
  \setlength{\labelsep}{1ex}%
}} {\end{list}\vspace{2pt}}

\begin{document}

\title{Fooling Functions of Halfspaces under Product Distributions}

\author{Parikshit Gopalan\\
Microsoft Research SVC\\
\texttt{parik@microsoft.com}
\and
Ryan O'Donnell\thanks{Work was partially done while the author
  consulted at Microsoft Research SVC. Supported by NSF grants CCF-0747250 and CCF-0915893, BSF grant 2008477, and Sloan and Okawa fellowships.}\\
Carnegie Mellon University\\
\texttt{odonnell@cs.cmu.edu}
\and
Yi Wu\thanks{Work done while an intern at Microsoft Research SVC.} \\
Carnegie Mellon University\\
\texttt{yiwu@cs.cmu.edu}\\
\and
David Zuckerman\thanks{Work was partially done while the author consulted at Microsoft Research SVC. 
Partially supported by NSF Grants CCF-0634811
and CCF-0916160 and THECB ARP Grant 003658-0113-2007.}\\
UT Austin\\
\texttt{diz@cs.utexas.edu}\\
}

\maketitle

\setcounter{page}{0}

\maketitle

\thispagestyle{empty}
\begin{abstract}

We construct pseudorandom generators that fool functions of halfspaces (threshold functions)
under a very broad class of product distributions. This class includes not only familiar cases such as the uniform distribution on the discrete cube, the uniform distribution on the solid cube, and the multivariate Gaussian distribution, but also includes any product of discrete distributions with probabilities bounded away from $0$.

Our first main result shows that a recent pseudorandom generator construction of Meka and
Zuckerman~\cite{MZ09}, when suitably modified, can fool arbitrary
functions of $d$ halfspaces under product distributions where each
coordinate has bounded fourth moment. To $\eps$-fool any size-$s$,
depth-$d$ decision tree  of
halfspaces, our pseudorandom generator uses 
seed length $O((d \log(ds/\eps) + \log n)\cdot \log(ds/\eps))$.
For monotone functions of $d$ halfspaces, the
seed length can be improved to $O((d \log(d/\eps) + \log n)\cdot\log(d/\eps))$.
We get better bounds for larger $\eps$; for example, to $1/\polylog(n)$-fool
all monotone functions of $(\log n)/\log\log n$ halfspaces, our generator
requires a seed of length just $O(\log n)$.

Our second main result generalizes the work of Diakonikolas et al.\ \cite{DGJSV:09} to show that bounded independence suffices to
fool functions of halfspaces under product distributions.
Assuming each coordinate satisfies a certain stronger moment condition,  
we show that any function computable by a size-$s$, depth-$d$ decision tree
of halfspaces is $\eps$-fooled by $\tilde{O}(d^4s^2/\eps^2)$-wise
independence. 

Our technical contributions include: a new multidimensional version of the classical Berry-Esseen theorem; a derandomization thereof; a generalization of Servedio~\cite{Ser:07}'s regularity lemma for halfspaces which works under any product distribution with bounded fourth moments; an extension of this regularity lemma to functions of many halfspaces; and, new analysis of the sandwiching polynomials technique of Bazzi~\cite{Baz09} for arbitrary product distributions.

\end{abstract}

\newpage

\section{Introduction}
\newcommand{\mc}[1]{\ensuremath{\mathcal{#1}}}
\newcommand{\mb}[1]{\ensuremath{\mathbb{#1}}}
\newcommand{\TCO}{\ensuremath{\mathsf{TC^0}}}
\newcommand{\BPP}{\ensuremath{\mathsf{BPP}}}
\newcommand{\PTIME}{\ensuremath{\mathsf{P}}}

\emph{Halfspaces}, or threshold functions, are a central class of Boolean-valued functions.
A halfspace is a function $h: \R^n \rightarrow \{0,1\}$ of the form
$h(x_1, \dots, x_n) = \Ind{w_1x_1 + \cdots + w_nx_n \geq \theta}$
where the weights $w_1,\ldots,w_n$ and the threshold $\theta$ are
arbitrary real numbers. These functions have been studied extensively in theoretical computer
science, social choice theory, and machine learning. In computer
science, they were first studied in the context of switching circuits;
see for instance~\cite{Dertouzos:65,Hu:65,LewisCoates:67,Sheng:69,Muroga:71}.   Halfspaces (with non-negative weights) have also been studied
extensively in game theory and social choice theory as models for
voting; see e.g.~\cite{Penrose:46,Isbell:69,DubeyShapley:79,TaylorZwicker:92}.  Halfspaces are also ubiquitous in machine learning contexts, playing a key role in many important algorithmic techniques, such as Perceptron , Support Vector Machine, Neural Networks, and AdaBoost.   One of the outstanding open problems in circuit lower bounds is to
find an explicit function that cannot be computed by a depth two
circuit (``neural network'') of threshold gates~\cite{HMP+:87,Krause:91,KrauseWaack:91,FKL+:01}.

In this work we investigate the problem of constructing explicit \emph{pseudorandom generators} for functions of halfspaces.  
\begin{definition}
A function $G: \{0,1\}^s \to B$ is a pseudorandom generator (PRG) with seed length $s$ and error $\epsilon$ for a class $\calF$ of functions from $B$ to $\zo$
under distribution $\calD$ on $B$  --- or more succinctly, $G$ $\epsilon$-fools $\calF$ under $\calD$ with seed length $s$ ---
if for all $f \in \calF$,
\[\Bigl|\Pr_{\bX \sim \calD}\,[f(\bX) = 1] - \Pr_{\bY \sim \{0,1\}^s}\,[f(G(\bY)) = 1]\,\Bigr| \leq \epsilon.\]
\end{definition}

Under the widely-believed complexity-theoretic assumption $\BPP =
\PTIME$, there must be a deterministic algorithm that can approximate
the fraction of satisfying assignments to any polynomial-size circuit
of threshold gates.  Finding such an algorithm even for simple
functions of halfspaces has proven to be a difficult derandomization
problem.  Very recently, however, there has been a burst of progress
on constructing PRGs for halfspaces~\cite{RS:08,DGJSV:09,MZ09}.  The
present paper makes progress on this problem in several different
directions, as do several concurrent and independent
works~\cite{HKM:09,DKN:09,BLY:09}. 

This flurry of work on PRGs for functions of halfspaces has several
motivations beyond its status as a fundamental derandomization task.
For one, it can be seen as a natural geometric problem, with
connections to deterministic integration; for instance, the problem of
constructing PRGs for halfspaces under the uniform distribution on the $n$-dimensional sphere
amounts to constructing a $\poly(n)$-sized set that hits every spherical cap with
roughly the right frequency~\cite{RS:08}.
Second, PRGs for halfspaces have applications in streaming
algorithms~\cite{GR:09}, while PRGs for functions of halfspaces can be
used to derandomize the Goemans-Williamson Max-Cut algorithm,
algorithms for approximate counting, algorithms for dimension
reduction and intractability results in computational learning~\cite{KS09}. Finally, proving lower
bounds for the class $\TCO$ of small depth threshold circuits is an outstanding open
problem in circuit complexity. An explicit PRG for a class
is easily seen to imply lower bounds against that class. Constructions
of explicit PRGs might shed light on structural properties of threshold
circuits and the lower bound problem.

\subsection{Previous work}
\newcommand{\etal}{{et al.\ }}

The work of Rabani and Shpilka~\cite{RS:08}  constructed a hitting set generator for halfspaces
under the uniform distribution on the sphere. Diakonikolas \etal \cite{DGJSV:09} constructed the first PRG
for halfspaces over bits; i.e., the uniform distribution on $\pmo^n$. They showed
that any $k$-wise independent distribution $\eps$-fools halfspaces with respect to the uniform
distribution for $k = \tilde{O}(1/\eps^2)$, giving PRGs
with seed length $(\log n)\cdot \tilde{O}(1/\eps^2)$. 

Meka and Zuckerman constructed a pseudorandom generator that $\eps$-fools degree-$d$ polynomial threshold 
functions (``PTFs'', a generalization of halfspaces) over uniformly random bits with
seed length $(\log n)/\eps^{O(d)}$ \cite{MZ09}. Their generator is a simplified
version of Rabani and Shpilka's hitting set generator. In the case of halfspaces, they combine their generator
with generators for small-width branching programs due to Nisan and Nisan-Zuckerman
\cite{Nis,NZ} to bring the seed length down to $O((\log n) \log(1/\eps))$.
This is the only previous or independent work where the seed length depends logarithmically on~$1/\eps$.

\subsection{Independent concurrent work}
Independently and concurrently, a number of   other researchers have extended some of the aforementioned results, mostly to intersections of halfspaces and polynomial threshold functions over the hypercube or Gaussian space.

Diakonikolas \etal~\cite{DKN:09} showed that $O(1/\eps^{9})$-wise independence
suffices to fool degree-2 PTFs under the uniform distribution on the
hypercube and under the Gaussian distribution. They also prove
that $\poly(d,1/\eps)$-wise independence suffices to fool
intersections of $d$ degree-2 PTFs in these settings.

Harsha \etal~\cite{HKM:09} obtain a PRG that fools
intersections of $d$ halfspaces under the Gaussian distribution with
seed length $O((\log n)\cdot \poly(\log d,1/\eps))$.  They
obtain similar parameters for intersections of $d$ ``regular'' halfspaces
under the uniform distribution on $\pmo^n$ (a halfspace is regular if
all of its coefficients have small magnitude compared to their sum of squares).

Ben-Eliezer \etal~\cite{BLY:09} showed that roughly $\exp((d/\eps)^d)$-wise
independence $\eps$-fools degree-$d$ PTFs which depend
on a small number of linear functions.

\subsection{Our Results}

In this work, we construct pseudorandom generators for arbitrary
functions of halfspaces under (almost) arbitrary product
distributions. Our work diverges from previous work in making minimal
assumptions about the distribution we are interested in, and in
allowing general functions of halfspaces. For both of our main
results, we only assume that the distribution is a product
distribution where each coordinate satisfies some mild 
conditions on its moments. These conditions include most
distributions of interest, such as the Gaussian distribution,
the uniform distribution on the hypercube, the uniform distribution on
the solid cube, and discrete distributions with probabilities bounded away from~$0$.  
Our results can also be used to fool the uniform distribution on the
sphere, even though it is not a product distribution. This allows us to
derandomize the hardness result of Khot and Saket~\cite{KS09} for learning
intersections of halfspaces.

We also allow for arbitrary functions of $d$ halfspaces, 
although the seed length improves significantly if we
consider monotone functions or small decision
trees. In particular, we get strong results for
intersections of halfspaces.  

\subsubsection{The Meka-Zuckerman Generator}

We show that a suitable modification of the Meka-Zuckerman (MZ) generator
can fool arbitrary functions of $d$ halfspaces under any product
distribution, where the distribution on each coordinate has
 bounded fourth moments. More precisely, we consider product distributions on $\bX =
 (\bx_1,\ldots,\bx_n)$ where for every $i \in [n]$,  $\E[\bx_i] = 0, \E[\bx^2_i] =1, \ \E[\bx^4_i]  \leq C$
where $C \geq 1$ is a parameter of the generator $G$. We say that
the distribution $\bX$ has $C$-bounded fourth moments.

We get our best results for monotone functions of $d$ halfspaces, such as intersections of $d$ halfspaces.
For distributions with polynomially bounded fourth moments,
our modified MZ PRG fools the intersection of $d$ halfspaces with polynomially small error using a seed of length $O(d\log^2 n)$.
Many natural distributions have $O(1)$-bounded fourth moments.  Even for $\polylog(n)$-bounded fourth moments,
our PRG
fools the intersection of $(\log n)/\log\log n$ halfspaces with error $1/\polylog(n)$ using a seed of length just $O(\log n)$.
Both of these cases are captured in the following theorem.

\begin{theorem}
\label{thm:main-1}
Let $\bX$ be sampled from a product distribution on $\mathbb{R}^n$ with $C$-bounded
fourth moments. 
The modified MZ generator
$\eps$-fools any monotone function of $d$ halfspaces with seed length
$O((d \log(Cd/\eps) + \log n)\log(Cd/\eps))$.
When $Cd/\eps \geq \log^{-c} n$ for any $c > 0$, the seed length becomes
$O(d \log(Cd/\eps) + \log n)$.
\end{theorem}

As a corollary, we get small seed length for functions of
halfspaces that have small decision tree complexity.
In the theorem below we could even take $s$ to be the minimum of the number of 0-leaves and 1-leaves.

\begin{theorem}
\label{thm:main-2}
Let $\bX$ be as in Theorem \ref{thm:main-1}. The modified MZ generator $\eps$-fools
any size-$s$, depth-$d$ function of halfspaces,
using a seed of length
$O((d \log(Cds/\eps) + \log n)\log(Cds/\eps))$.
When $Cds/\eps \geq \log^{-c} n$ for any $c>0$, the seed length becomes
$O(d \log(Cds/\eps) + \log n)$.
\end{theorem}

Since the decision tree complexity is at most $2^d$, we deduce the following.

\begin{corollary}
Let $\bX$ be as in theorem \ref{thm:main-1}. The modified MZ generator
$\eps$-fools any function of $d$~halfspaces, using a seed of length
$O((d^2 +d \log(Cd/\eps) + \log n)(d+\log(Cd/\eps)))$.
When $Cd2^d/\eps \geq \log^{-c} n$ for any $c>0$, the seed length becomes
$O(d^2 + d \log(Cd/\eps) + \log n)$.
\end{corollary}

\subsubsection{Bounded Independence fools functions of halfspaces} 

We prove that under a large class of product distributions, bounded
independence suffices to fool functions of $d$ halfspaces. This
significantly generalizes the result of Diakonikolas
\etal~\cite{DGJSV:09} who proved that bounded independence fools halfspaces under the
uniform distribution on $\pmo^n$.  The condition necessary on the product distributions is unfortunately somewhat technical; we state here a theorem that covers the main cases of interest:

\begin{theorem} \label{thm:main-dwise} Suppose $f$ is computable as a size-$s$, depth-$d$ function of halfspaces over the independent random variables $\bx_1, \dots, \bx_n$. If we assume the $\bx_j$'s are discrete, then $k$-wise independence suffices to $\eps$-fool $f$, where 
\[
k = \wt{O}(d^4 s^2 / \eps^2) \cdot \poly(1/\alpha).
\]
Here $0 < \alpha \leq 1$ is the least nonzero probability of any outcome for an $\bx_j$.  Moreover, the same result holds with $\alpha = 1$ for certain continuous random variables $\bx_j$, including Gaussians (possibly of different variance) and random variables which are uniform on (possibly different) intervals.
\end{theorem}
For example, whenever $\alpha \geq 1/\polylog(d/\eps)$ it holds that
$\wt{O}(d^6/\eps^2)$-wise independence suffices to $\eps$-fool
intersections of $m$ halfspaces.  For random variables that do not
satisfy the hypotheses of Theorem~\ref{thm:main-dwise}, it may still
be possible to extract a similar statement from our techniques.
Roughly speaking, the essential requirement is that the random
variables $\bx_j$ be ``$(p, 2, p^{-c})$-hypercontractive'' for large values of $p$ and some constant $c < 1$.  

\newcommand{\cmatrix}{M}

\paragraph{\bf Notation: }
Throughout, all random variables take values in $\R$ or $\R^d$. 
Random variables will be in boldface.  Real scalars will be lower-case
letters; real vectors will be upper-case letters.  If $X$ is a
$d$-dimensional vector, we will write $X[1], X[2], \dots, X[d]$ for
its coordinates values and $\len{X} = \sqrt{\sum_{i=1}^d X[i]^2}$ for
its Euclidean length. When $\cmatrix$ is a matrix, we also use the notation
$\cmatrix[i,j]$ for its $(i,j)$ entry.  If $\bX$ is a vector-valued random
variable, we write $\rnorm{\bX}_p = \E[\len{\bX}^p]^{1/p}$.  We
typically use $i$ to index dimensions and $j$ to index sequences.
Given $x \in \R$ we define $\sgn(x) = 1$ if $x \geq 0$ and $\sgn(x) =
-1$ if $x < 0$.  If $X$ is a $d$-dimensional vector, then $\vsgn(X)$
denotes the vector in $\{-1,1\}^d$ with $\vsgn(X)[i] = \sgn(X[i])$. 

Our results concern arbitrary functions of $d$ halfspaces.
Thus we have vectors $W_1,\ldots, W_n,\Theta \in \R^d$, and we're interested in functions $f : \{-1,1\}^d \to \{0,1\}$ of the vector
$\vsgn(x_1W_1 + \ldots + x_nW_n - \Theta)$, which we abbreviate to
$\vsgn(W\cdot X - \Theta)$ where $W = (W_1,\ldots,W_n)$ and $X =(x_1,\ldots,x_n)$.

\paragraph{\bf Organization: }We give an overview of our results and
their proofs in \ref{sec:overview}. We prove the multi-dimensional
Berry-Esseen type theorems in Section \ref{sec:berry-e}. In Section
\ref{sec:crit}, we prove a regularity lemma for multiple halfspaces in the
general setting of hypercontractive variables. We state modified MZ
generator in Section \ref{ap:mzg}, and analyze it using the machinery
above in Section \ref{ap:mzg-analysis}. In Section \ref{ap:monotone},
we show how to combine it with PRGs for branching programs to get our
Theorems \ref{thm:main-1} and \ref{thm:main-2}. We prove Theorem
\ref{thm:main-dwise} in Section \ref{ap:bounded}. In Section
\ref{ap:sphere}, we show how our results apply to fooling the uniform
distribution on the sphere, and use it to derandomize the hardness
result of \cite{KS09}.

\section{Overview of the main results}
\label{sec:overview}
In this section, we give an overview  on how we construct and analyze the following two types of PRGs for functions of halfspaces under general product distributions: i) the modified Meka-Zuckerman generator (in Section~\ref{sec:mzg}) and ii) the bounded  independence generator (in Section~\ref{sec:bdg})
\subsection{The Meka-Zuckerman Generator}\label{sec:mzg}
There are five steps in the analysis:
\begin{tightenumerate}
\item Discretize the distribution $\bX$ so that it is the
  product of discrete distributions whose moments nearly match those
  of $\bX$.
\item Prove a multidimensional version of the classical Berry-Esseen
  theorem, and a derandomization thereof under general product distributions. This allows us to
  handle functions of regular halfspaces. See Subsection~\ref{subsec:berry-e}.
\item Generalize the regularity lemma/critical index
  lemma (see \cite{Ser:07,DGJSV:09}) to $d$
  halfspaces under general product distributions. This gives a small set of variables such that after conditioning on
  these variables, each halfspace becomes either regular or close to a
  constant function. See Subsection~\ref{subsec:crit}.
\item Use the regularity lemma to reduce analyzing functions of
  $d$ arbitrary halfspaces to analyzing functions of $d$ (or fewer) regular
  halfspaces.
\item Finally, generalize the monotone trick from \cite{MZ09}, which previously worked only for a single ``monotone'' branching program,
 to monotone functions of monotone branching programs. This enables us
 to get seed length logarithmic in $1/\eps$.
 See Subsection~\ref{subsec:mono}.
\end{tightenumerate}

\subsubsection{Multi-Dimensional Berry-Esseen Theorem}
\label{subsec:berry-e}

The classic Berry-Esseen Theorem is a quantitative version of the Central Limit Theorem.
This theorem is essential in the analyses of \cite{MZ09} and \cite{DGJSV:09} for one halfspace.
Since we seek to fool functions of several halfspaces, we prove a
multi-dimensional version of the Berry-Esseen theorem, which approximates
the distribution of $\sum_i\bx_iW_i$. The error of the approximation
is small when all the halfspaces are regular (no coefficient is
too large). While there are multi-dimensional versions known,
we were unable to find in the literature any theorems which we could use in a
``black-box'' fashion.  The reason for this is twofold: known results
tend to focus on measuring the difference between
probability distributions vis-a-vis convex sets; whereas, we are
interested in more specialized sets, unions of orthants.  Second,
results in the literature tend to assume a nonsingular covariance
matrix and/or have a dependence in the error bound on its least
eigenvalue; whereas, we need to work with potentially singular
covariance matrices. We believe this theorem could be of independent interest.

Next we show how this theorem can be derandomized in a certain sense.
This derandomization enables us to show that our modified MZ PRG fools
\emph{regular} halfspaces.  

\subsubsection{Multi-Dimensional Critical Index}
\label{subsec:crit}

The concept of critical index was introduced in the work of Servedio
\cite{Ser:07}. It is used to prove a regularity lemma for
  halfspaces, which asserts that every halfspace contains a {\em
  head} consisting of constantly many variables, such that once these
variables are set randomly, the resulting function is either close to
constant, or close to a regular halfspace. This lemma has found
numerous applications in complexity and learning theoretic questions
related to halfspaces \cite{Ser:07,OS:08,FGRW:09,DGJSV:09,MZ09}.

The obvious generalization of the one-dimensional theorem to multiple
halfspaces would be to take the union of the heads of each halfspace. This
does not work, since setting variables in a regular halfspace
can make it irregular.  We prove a multidimensional version of this
lemma, which moreover holds in the setting of product distributions
with bounded fourth moments. Our analysis shows  that the lemma only requires some
basic concentration and anti-concentration properties, which are
enjoyed by any random variable with bounded fourth moments.

\subsubsection{Monotone Branching Programs}
\label{subsec:mono}

The only known method to get logarithmic dependence on $1/\eps$ for
PRGs for halfspaces, due to Meka and Zuckerman, considers the natural 
branching program accepting a halfspace.  
This branching program is ``monotone,''
in the sense that in every layer the set of accepting suffixes forms
a total order under inclusion.  
Meka and Zuckerman showed that any monotone branching program of
arbitrary width can be sandwiched between two small-width monotone
branching programs.
Therefore, PRGs for small-width branching programs, such as those by
Nisan~\cite{Nis} can be used.

Since we deal with several halfspaces, we get several monotone branching
programs.  We consider monotone functions of monotone branching programs,
to encompass intersections of halfspaces.  However, such functions are
not necessarily computable by monotone branching programs.  Nevertheless,
we show how to sandwich such functions between two small-width branching
programs, and thus can use the PRGs like Nisan's.

\subsection{Bounded Independence fools functions of halfspaces} \label{sec:bdg}

\subsubsection{Sandwiching ``polynomials''} To prove that bounded independence can fool functions of halfspaces (Theorem~\ref{thm:main-dwise}), we use the ``sandwiching polynomials'' method as introduced by Bazzi~\cite{Baz09} and used by~\cite{DGJSV:09}.  However in our setting of general random variables it is not appropriate to use polynomials per se.  The essence of the sandwiching polynomial method is showing that only groups of $d$ random variables need to be ``simultaneously controlled'.  When the random variables are $\pm 1$-valued, controlling sub-functions of at most $d$ random variables is equivalent to controlling polynomials of degree at most $d$.  But for random variables with more than two outcomes, a function of $d$ random variables requires degree higher than $d$ in general, a price we should not be forced to pay.  We instead introduce the following notions:

\begin{definition} Let $\Omega = \Omega_1 \times \cdots \times \Omega_n$ be a product set.  We say that $p : \Omega \to \R$ is a \emph{$k$-junta} if $f(x_1, \dots, x_n)$ depends on at most $k$ of the $x_j$'s.  We say that $p$ is a \emph{generalized polynomial of order (at most) $k$} if it is expressible as a sum of simple functions of order at most $k$.  In the remainder of this section we typically drop the word ``generalized'' from ``generalized polynomial'', and add the modifier ``ordinary'' when referring to ``ordinary polynomials''.
\end{definition}

We now give the simple connection to fooling functions with bounded independence:
\begin{definition} \label{def:sandwich} Let $\bX = (\bx_1, \ldots, \bx_n)$ be a vector of independent random variables, where $\bx_j$ has range $\Omega_j$.  Let $f : \Omega \to \R$, where $\Omega = \Omega_1 \times \cdots \times \Omega_n$.  We say that polynomials $p_l, p_u : \Omega \to \R$ are \emph{$\eps$-sandwiching} for $f$ if
\begin{gather*}
p_l(X) \leq f(X) \leq p_u(X) \text{ for all $X \in \Omega$, and }
\E[p_u(\bX)] -\eps \leq \E[f(\bX)] \leq \E[p_l(\bX)] + \eps.
\end{gather*}
\end{definition}
\begin{proposition} \label{prop:sandwiching-fools}  Suppose $p_l$, $p_u$ are $\eps$-sandwiching for $f$ as in Definition~\ref{def:sandwich} and have order at most $k$.  Then $f$ is $\eps$-fooled by $k$-wise independence.  I.e., if $\bY = (\by_1, \dots, \by_n)$ is a vector of random variables such that each marginal of the form $(\by_{j_1}, \dots, \by_{j_k})$ matches the corresponding marginal $(\bx_{j_1}, \dots, \bx_{j_k})$, then 
\[
\left|\E[f(\bX)] - \E[f(\bY)]\right| \leq \eps.
\]
\end{proposition}
\begin{proof} Write $p_u = \sum_{t} q_t$, where each $q_t$ is a $k$-junta.  Then
\[
\E[f(\bY)] \leq \E[p_u(\bY)] = \E[\littlesum_t q_t(\bY)] = \littlesum_t \E[q_t(\bY)] = \littlesum_t \E[q_t(\bX)] = \E[p_u(\bX)] \leq \E[f(\bX)] + \eps,
\]
where in addition to the sandwiching properties of $p_u$ we used the fact that $q_t$ is a $k$-junta to deduce $\E[q_t(\bY)] = \E[q_t(\bX)]$.  We obtain the bound $\E[f(\bY)] \geq \E[f(bX)] - \eps$ similarly, using $p_l$.
\end{proof}

\subsubsection{Upper polynomials for intersections suffice} 
We begin with a trivial observation:
\begin{proposition} \label{prop:upper-suffices}
Let $\calC$ be a class of functions $\Omega \to \{0,1\}$, and suppose that for every $f \in \calC$ we have just the ``upper sandwiching polynomial'', $p_u$, of an $\eps$-sandwiching pair for $f$.  Then if $\calC$ is closed under Boolean negation, we obtain a  matching ``lower polynomial'' $p_l$ of the same order as $p_u$ automatically.  
\end{proposition}
This is simply because given $p_u$ for $f$, we may take $p_l = 1 -
p_u$.  Since the Boolean negation of a halfspace is a halfspace, this observation could have been used for slight simplification in~\cite{DGJSV:09}.

Our Theorem~\ref{thm:main-dwise} is concerned with the class of $0$-$1$ functions $f$ computable as size-$s$, depth-$d$ functions of halfspaces.  This class is closed under Boolean negation; hence it suffices for us to obtain upper sandwiching polynomials.  Furthermore, every such $f$ can be written as
$
f = \sum_{t = 1}^{s'} H_t,
$
where $s' \leq s$ and $H_t$ is an intersection (AND) of up to $d$ halfspaces.  To see this, simply sum the indicator function for each root-to-leaf path in the decision tree (this again uses the fact that the negation of a halfspace is a halfspace).  Thus if we have  $(\eps/s)$-sandwiching upper polynomials of order $k$ for each $H_t$, by summing them we obtain an $\eps$-sandwiching upper polynomial for $f$ of the same order.  Hence to prove our main Theorem~\ref{thm:main-dwise}, it suffices to prove the following:

\begin{theorem} \label{thm:main-upper} Suppose $f$ is the intersection of $d$ halfspaces $h_1, \dots, h_d$ over the independent random variables $\bx_1, \dots, \bx_n$. Suppose $\alpha$ is as in Theorem~\ref{thm:main-dwise}.  Then there exists an $\eps$-sandwiching upper polynomial for $f$ of order $k \leq \wt{O}(d^4 / \eps^2) \cdot \poly(1/\alpha)$.
\end{theorem}

\subsubsection{Polynomial construction techniques}
Suppose for simplicity we are only concerned with the intersection $f$ of $d$ halfspaces $h_1, \dots, h_d$ over uniform random $\pm 1$ bits $\bx_j$.  The work of Diakonikolas et al.~\cite{DGJSV:09} implies that there are is an $\eps_0$-sandwiching upper polynomial $p_i$ of order $\wt{O}(1/\eps_0^2)$ for each $h_i$.  To obtain an $\eps$-sandwiching upper polynomial for the intersection $h_1 h_2 \cdots h_d$, a natural first idea is simply to try $p = p_1 p_2 \cdots p_{d}$.  This is certainly an upper-bounding polynomial; however the $\eps$-sandwiching aspect is unclear.  We can begin the analysis as follows.  Let $\bh_i = h_i(\bX)$ and $\bp_i = p_i(\bX)$.  By telescoping, 
\begin{eqnarray} 
\E[\bp_1 \cdots \bp_d] - \E[\bh_1 \cdots \bh_d] &=& \E[(\bp_1- \bh_1)\bp_2 \cdots \bp_d] \quad+\quad \cdots \nonumber\\
\dots &+& \E[\bh_1 \cdots \bh_{i-1} (\bp_i - \bh_i)\bp_{i+1} \cdots \bp_d] \quad+\quad \cdots \label{eqn:hybrid}\\
\dots &+& \E[\bh_1 \cdots \bh_{d-1}(\bp_d - \bh_d)]. \nonumber
\end{eqnarray}
Now the last term here could be upper-bounded as
\[
\E[\bh_1 \cdots \bh_{d-1}(\bp_d - \bh_d)] \leq \E[\bp_d - \bh_d] \leq \eps_0,
\]
since each $0 \leq \bh_i \leq 1$ with probability~$1$.  But we cannot make an analogous bound for the remaining terms because we have no a priori control over the values of the $\bp_i$'s beyond the individual sandwiching inequalities
\[
\E[\bp_i - \bh_i] \leq \eps_0.
\]
Nevertheless, we will be able to make this strategy work by establishing \emph{additional boundedness conditions} on the polynomials $p_i$; specifically, that each $\bp_i$ exceeds $1 + 1/d^2$ extremely rarely, and that even the high $2d$-norm of $\bp_i$ is not much more than~$1$.\\

Establishing these extra properties requires significant reworking the
construction in~\cite{DGJSV:09}.  Even in the case of uniform random
$\pm 1$ bits, the calculations are not straightforward, since the
upper sandwiching polynomials implied by~\cite{DGJSV:09} are only
fully explicit in the case of regular halfspaces.  And to handle
general random variables $\bx_j$, we need more than just our new
Regularity Lemma~\ref{thm:crit-one} for halfspaces.  We also need to assume a
stronger hypercontractivity property of the random variables to ensure
they have rapidly decaying tails.

\section{Hypercontractivity}
\label{sec:hyper}

The notion of of \emph{hypercontractive random variables} was
introduced in~\cite{KS88} and developed by Krakowiak, Kwapie{\'n}, and
Szulga:

\begin{definition} We say that a real random variable $\bx$ is
  \emph{$(p,q,\eta)$-hypercontractive} for $1 \leq q \leq p < \infty$
  and $0 < \eta < 1$ if $\rnorm{\bx}_p < \infty$, and for all $a \in \R$,
$\rnorm{a + \eta \bx}_p \leq \rnorm{a + \bx}_q$.
\end{definition}

In this paper we will be almost exclusively concerned with the
  simplest case, $p = 4$, $q = 2$.  Let us abbreviate the definition
  in this case (and also exclude constantly-$0$ random variables): 
\begin{definition} A real random variable $\bx$ is \emph{$\eta$-HC}
  for $0 < \eta < 1$ if $0 < \rnorm{\bx}_4 < \infty$ and for all $a
  \in \R$, $\rnorm{a + \eta \bx}_4 \leq \rnorm{a + \bx}_2$, i.e. $\E[(a+\eta \bx)^4] \leq \E[(a+\bx)^2]^2$.
\end{definition}

Essentially, a mean $0$ real random variable is $\eta$-HC with large
$\eta$ if and only if it has a small $4$th moment (compared to its
$2$nd moment).  Random variables with small $4$th
moment are known to enjoy some basic concentration and anti-concentration
properties.  We work with hypercontractivity rather than $4^{th}$
moments because it tends to slightly shorten proofs and improve
constants; the main convenience is that a linear combination of
$\eta$-HC random variables is also $\eta$-HC.  

Here we list some basic and useful properties of $\eta$-HC random variables,
all of which have elementary proofs.  Note that Facts~\ref{4norm}
and~\ref{converse} imply that the upper bound on the $4$th norm
$C=\Theta(1/\eta^4)$.
\begin{fact} \cite{KS88,MOO05,Wol06,Wol07}
\label{fact:hyper}
\begin{tightenumerate}
\item If $\bx$ is $\eta$-HC then it is also $\eta'$-HC for all $\eta' < \eta$. 
\item If $\bx$ is $\eta$-HC then $\bx$ is centered, $\E[\bx] = 0$.
\item 
\label{4norm}
If $\bx$ is $\eta$-HC then $\E[\bx^4] \leq (1/\eta)^4 \E[\bx^2]^2$.  
\item
\label{converse}
Conversely, if $\E[\bx] = 0$ and $\E[\bx^4] \leq (1/\eta)^4 \E[\bx^2]^2$, then $\bx$ is $(\eta/2\sqrt{3})$-HC.  If $\bx$ is also \emph{symmetric} (i.e., $-\bx$ has the same distribution as $\bx$) then $\bX$ is $\min(\eta, 1/\sqrt{3})$-HC.
\item If $\bx$ is $\pm 1$ with probability $1/2$ each, then $\bx$ is $(1/\sqrt{3})$-HC.  The same is true if $\bx$ has the standard Gaussian distribution or the uniform distribution on $[-1,1]$. 
\item If $\bx$ is $\eta$-HC then in fact $\eta \leq 1/\sqrt{3}$.
\item If $\bx$ is a centered discrete random variable and $\alpha \leq 1/2$ is the least nonzero value of $\bx$'s probability mass function, then $\bx$ is $\eta$-HC for $\eta = \alpha^{1/4}/2\sqrt{3}$.
\item If $\bx_1, \dots, \bx_n$ are independent $\eta$-HC random variables, then so is $c_1 \bx_1 + \cdots c_n \bx_n$ for any real constants $c_1, \dots, c_n$, not all $0$.  (Indeed, $4$-wise independence suffices.)
\item If $\bx$ is $\eta$-HC, and $\by$ is a random variable with the
  same $r$th moments as $\bx$ for all $r = 0, 1, 2, 3, 4$, then $\by$
  is also $\eta$-HC. 
\end{tightenumerate}
\end{fact}

The notion of hypercontractivity can be extended to $\R^d$-valued random variables:
\begin{definition} An $\R^d$-random variable $\bX$ is \emph{$\eta$-HC}
  for $0 < \eta < 1$ if $\rnorm{\bX}_4 < \infty$ and for all $A \in
  \R^d$, $
\rnorm{A + \eta \bX}_4 \leq \rnorm{A + \bX}_2$.
\end{definition}

We require the following facts about vector-valued hypercontractivity:
\begin{fact} \cite{Szu90}
\begin{tightenumerate}
\item If $W \in \R^d$ is a fixed vector and $\bx$ is an $\eta$-HC real random variable, then $\bX = \bx W$ is an $\eta$-HC.
\item If $\bX_1, \dots, \bX_n$ are independent $\eta$-HC random
  vectors, then so is $c_1 \bX_1 + \cdots c_n \bX_n$ for any real
  constants $c_1, \dots, c_n$.  (Again, $4$-wise independence also
  suffices.) 
\end{tightenumerate}
\end{fact}


Hypercontractive real random variables possess the following good concentration and anti-concentration properties.
\begin{proposition}  
\label{prop:HC-conc} If $\bx$ is $\eta$-HC then for all $t > 0$, 
$\Pr[\abs{\bx} \geq t\rnorm{\bx}_2] \leq \frac{1}{\eta^4 t^4}$.
\end{proposition}
\begin{proof}
Apply Markov to the event ``$\bx^4 \geq t^4\E[\bx^2]^2$''.
\end{proof}

\begin{proposition}  
\label{prop:HC-anticonc} If $\bx$ is $\eta$-HC then for all $\theta \in \R$ and $0 < t < 1$, 
$\Pr[\abs{\bx - \theta} > t\rnorm{\bx}_2] \geq \eta^4 (1-t^2)^2$.
\end{proposition}
\begin{proof}
By scaling $\bx$ it suffices to consider the case $\rnorm{\bx}_2 = 1$.  Consider the random variable $\by = (\bx - \theta)^2$.  We have 
\[
\E[\by] = \E[\bx^2] - 2 \theta \E[\bx] + \theta^2 = 1 + \theta^2,
\]
\[
\E[\by^2] = \eta^{-4} \E[(-\eta \theta + \eta \bx)^4] \leq \eta^{-4} \E[(-\eta \theta + \bx)^2]^2 = \eta^{-4}(1 + \eta^2 \theta^2)^2 = (\eta^{-2} + \theta^2)^2,
\]
where we used the fact that $\bx$ is $\eta$-HC in the second calculation (and then used the first calculation again).  We now apply the Paley-Zygmund inequality (with parameter $0 < t^2/(1+\theta^2) < 1$):
\begin{multline} \label{eqn:anticoncratio}
\Pr[\abs{\bx - \theta} > t] = \Pr[\by > t^2] = \Pr\left[\by > \frac{t^2}{1+\theta^2}\E[\by]\right] \geq \left(1 - \frac{t^2}{1+\theta^2}\right)^2 \frac{\E[\by]^2}{\E[\by^2]} \\
 \geq \left(1 - \frac{t^2}{1+\theta^2}\right)^2 \frac{(1+\theta^2)^2}{(\eta^{-2} + \theta^2)^2} 
 = \left(\frac{\eta^2(1-t^2) + \eta^2\theta^2}{1 + \eta^2\theta^2}\right)^2.
\end{multline}
Treat $\eta$ and $t$ as fixed and $\theta$ as varying. Writing $u =
\eta^2(1-t^2)$, we have $0 < u < 1$; hence the fraction
$(u+\eta^2\theta^2)/(1+\eta^2\theta^2)$ appearing
in~\eqref{eqn:anticoncratio} is positive and increasing as $\eta^2
\theta^2$ increases.  Thus it is minimized when $\theta = 0$;
substituting this into~\eqref{eqn:anticoncratio} gives the claimed
lower bound. 
\end{proof}

\section{The Multi-Dimensional Berry-Esseen Theorem}
\label{sec:berry-e}

In this section we prove a Berry-Esseen-style results in the setting of
multidimensional random variables, and a derandomization of it. 

We assume the following setup:  $\bThm_1, \dots, \bThm_n$ are
independent $\R^d$-valued $\eta$-HC random variables, not necessarily
identically distributed, satisfying $\E[\bThm_j] = 0$ for all $j \in
[n]$. We let $\bS = \bThm_1 + \cdots + \bThm_n$.  We write $\cmatrix_j = \cov[\bThm_j] \in \R^{d \times d}$ for the covariance matrix of $\bThm_j$, which is positive semidefinite.  We also write $\cmatrix = \cov[\bS]$ for the covariance matrix of $\bS$; by the independence and mean-zero assumptions we have $\cmatrix = \cmatrix_1 + \cdots + \cmatrix_n$.  We will also assume that
\[
\cmatrix[i,i] = \sum_{j=1}^n \E\left[\bThm_j[i]^2\right] = 1 \qquad \text{for all $i \in [d]$.}
\]
If we write $\sigma_j^2 = \rnorm{\bThm_j}^2$, it follows that $\sum_{j=1}^n \sigma_j^2 = d$.
We introduce new independent random variables $\bG_1, \dots,
\bG_n$, where $\bG_j$ is a $d$-dimensional Gaussian random variable
with covariance matrix $\cmatrix_j$; we also write also $\bG = \bG_1 + \cdots
+ \bG_n$.  We say that $A \subseteq \R^d$ is a translate of a union of orthants
if there exists some vector $\Theta \in \R^d$ such that
$X \in A$ depends only on $\vsgn(X - \Theta)$.

\begin{theorem}
\label{thm:be}
Let $\bS$ and $\bG$ be as above. Let $A \subseteq \R^d$ be a translate of a union of orthants. Then
\[
\abs{\Pr[\bS \in A] - \Pr[\bG \in A]} \leq O(\eta^{-1/2} d^{13/8}) \cdot \Bigl(\littlesum_{j=1}^n \sigma_j^4\Bigr)^{1/8}.
\]
\end{theorem}

We now show that this result can be ``derandomized'' using the output
of the MZ generator $\bY$ in place of $\bX$. We describe here a simplified
version of the output of their generator. 

\newcommand{\collision}{b}

\begin{definition}
\label{def:hash-collision}
A family $\calH = \{h:[n] \to [t]\}$ of hash functions is  \emph{$\collision$-collision preserving} if
\begin{enumerate}
\item
For all $i \in [n], \ell \in [t]$, $\Pr_{h \in_u \calH}[h(i) = \ell ] \leq \collision/t$.
\item
For all $i \neq j \in [n]$,
$\Pr_{h \in_u \calH}[h(i) =  h(j) ] \leq \collision/t$.
\end{enumerate}
\end{definition}

Efficient constructions of size $|\calH| = O(nt)$ are known for any constant 
$\collision \geq 1$.
$\collision=1$ is optimal, and can be achieved by a pairwise independent
family. In our construction we use $\collision=1$, 
but we will need larger $\collision$ in our analysis. 
A hash function induces a partition of $[n]$.

We choose a partition $\bH_1, \dots, \bH_t$ of $[n]$ into $t$ buckets using a $\collision$-collision preserving family
of hash functions (where $\collision \leq 2$). The vector of variables
$\{\bY_j\}_{ j\in \bH_\ell}$ is generated $4$-wise independently. There
is full independence across different buckets. Let $\bT = \bDerand_1 + \cdots + \bDerand_n$. 

\begin{theorem}
\label{thm:be-derandomized}
Let $\bT$ and $\bG$ be as above. Let $A \subseteq \R^d$ be a translate of a union of orthants. Then
\[
\abs{\Pr[\bT \in A] - \Pr[\bG \in A]} \leq O(\eta^{-1/2} d^{13/8}) \cdot \Bigl(\frac{d^2}{t} + \littlesum_{j=1}^n \sigma_j^4\Bigr)^{1/8}.
\]
\end{theorem}

Putting these two theorems together, we have shown the following statement
\begin{theorem}
\label{thm:very-long}
Let $\bS$ and $\bT$ be as above. Let $A \subseteq \R^d$ be a translate of a union of orthants. Then
\[
\abs{\Pr[\bS \in A] - \Pr[\bT \in A]} \leq O(\eta^{-1/2} d^{13/8}) \cdot \Bigl(\frac{d^2}{t} + \littlesum_{j=1}^n \sigma_j^4\Bigr)^{1/8}.
\]
\end{theorem}

In the rest of the section, we prove above theorems; our aim is not to get the best bounds possible (for which one might
pursue the methods of Bentkus~\cite{Ben04}).  Rather, we aim to provide a simple method which achieves a reasonable bound, and thus
use the Lindeberg method, following~\cite{MOO05,Mos08} very closely.

\subsection{The basic lemma}

In what follows, $K$ will denote a $d$-dimensional multi-index $(k_1, \dots, k_d) \in \N^d$, with $\abs{K}$ denoting $j_1 + \cdots + j_d$ and $K!$ denoting $k_1! k_2! \cdots k_d!$.  Given a vector $H \in \R^d$, the expression $H^K$ denotes $\prod_{i=1}^d H[i]^{k_i}$. Given a function $\psi : \R^d \to \R$, the expression $\psi^{(K)}$ denotes the mixed partial derivative taken $k_i$ times in the $i$th coordinate; we will always assume $\psi$ is smooth enough that the order of the derivatives does not matter.

The following lemma is essentially proven in, e.g.,~\cite[Theorem~4.1]{Mos08}. To obtain it, simply repeat Mossel's proof in the degree~$1$ case, until equation~(31).  (Although Mossel assumes that the covariance matrices $\cmatrix_j$ are identity matrices, this is not actually necessary; it suffices that $\cov[\bThm_j] = \cov[\bG_j]$.) Then instead of using hypercontractivity, skip directly to summing the error terms over all coordinates.
\ignore{Honest.  I started writing the full proof, but it was so much just copying what's in Elchanan's paper that I stopped.}

\begin{lemma}  \label{lem:mossel} Let $\psi : \R^d \to \R$ be a $\calC^3$ function with $\abs{\psi^{(K)}} \leq \nice$ for all $\abs{K} =3$.  Then
\begin{equation} \label{eqn:mossel}
\abs{\E[\psi(\bS)] - \E[\psi(\bG)]} \leq \nice \sum_{\abs{K}=3} \frac{1}{K!} \sum_{j=1}^n \left(\E\left[\abs{\bThm_j^K}\right] + \E\left[\abs{\bG_j^K}\right]\right).
\end{equation}
\end{lemma}

We further deduce:
\begin{corollary} \label{cor:inv}  In the setting of Lemma~\ref{lem:mossel},
\[
\abs{\E[\psi(\bS)] - \E[\psi(\bG)]} \leq 2\nice d^3 \sum_{j=1}^n \rnorm{\bThm_j}_3^3.
\]
\end{corollary}
\begin{proof}
Fix a multi-index $K$ with $\abs{K} = 3$ and also an index $j$.  We will show that
\begin{equation} \label{eqn:corgoal}
\E\left[\abs{\bThm_j^K}\right] + \E\left[\abs{\bG_j^K}\right] \leq 2.6 \rnorm{\bThm_j}_3^3.
\end{equation}
Substituting this into~\eqref{eqn:mossel} completes the proof, since\rnote{check this please}
\[
\nice \sum_{\abs{K}=3} \frac{2.6}{K!} \leq 2\nice d^3.
\]

Let the nonzero coordinates in $K$ be $i_1, i_2, i_3 \in [d]$, written with multiplicity.  Write also
\[
\sigma_i^2 = \cmatrix_j[i,i] = \E\left[\bG_j[i]^2\right] = \E\left[\bThm_j[i]^2\right].
\]

On one hand, by H\"{o}lder we have
\[
\E\left[\abs{\bG_j^K}\right] = \E\left[\abs{\bG_j[i_1]\bG_j[i_2]\bG_j[i_3]}\right] \leq 
\sqrt[3]{\E\left[\abs{\bG_j[i_1]}^3\right]\E\left[\abs{\bG_j[i_2]}^3\right]\E\left[\abs{\bG_j[i_3]}^3\right]}.
\]
Note that the distribution of $\bG_j[i_1]$ is $N(0, \sigma_{i_1}^2)$.  It is elementary that such a random variable has third absolute moment equal to $2\sqrt{2/\pi} \cdot \sigma_{i_1}^{3} \leq 2.6\sigma_{i_1}^{3}$.  As the same is true for $i_2$ and $i_3$, we conclude that
\begin{equation} \label{eqn:gupper}
\E\left[\abs{\bG_j^K}\right] \leq 1.6 \sigma_{i_1} \sigma_{i_2} \sigma_{i_3}.
\end{equation}

On the other hand, we can similarly upper-bound
\begin{equation} \label{eqn:xupper}
\E\left[\abs{\bThm_j^K}\right] \leq 
\sqrt[3]{\E\left[\abs{\bThm_j[i_1]}^3\right]\E\left[\abs{\bThm_j[i_2]}^3\right]\E\left[\abs{\bThm_j[i_3]}^3\right]}
\end{equation}
But
\[
\sqrt[3]{\E\left[\abs{\bThm_j[i_1]}^3\right]\E\left[\abs{\bThm_j[i_2]}^3\right]\E\left[\abs{\bThm_j[i_3]}^3\right]} \geq \sqrt[3]{\E\left[\abs{\bThm_j[i_1]}^2\right]^{3/2}\E\left[\abs{\bThm_j[i_2]}^2\right]^{3/2}\E\left[\abs{\bThm_j[i_3]}^2\right]^{3/2}} = \sigma_{i_1} \sigma_{i_2} \sigma_{i_3},
\]
and hence from~\eqref{eqn:gupper} and~\eqref{eqn:xupper} we conclude
\[
\E\left[\abs{\bThm_j^K}\right] + \E\left[\abs{\bG_j^K}\right] \leq 2.6 \sqrt[3]{\E\left[\abs{\bThm_j[i_1]}^3\right]\E\left[\abs{\bThm_j[i_2]}^3\right]\E\left[\abs{\bThm_j[i_3]}^3\right]}.
\]
Finally, we clearly have $\abs{\bThm_j[i_1]} \leq \len{\bThm_j}$ always, and similarly for $j_2$, $j_3$.  Hence
\[
\E\left[\abs{\bThm_j^K}\right] + \E\left[\abs{\bG_j^K}\right] \leq 2.6 \sqrt[3]{\E\left[\len{\bThm_j}^3\right]\E\left[\len{\bThm_j}^3\right]\E\left[\len{\bThm_j}^3\right]} = 2.6\rnorm{\bThm_j}_3^3,
\]
confirming~\eqref{eqn:corgoal}.
\end{proof}

\begin{corollary} \label{cor:fourth}  In the setting of Lemma~\ref{lem:mossel},
\[
\abs{\E[\psi(\bS)] - \E[\psi(\bG)]} \leq 2\nice d^{7/2} \sqrt{\sum_{j=1}^n \rnorm{\bThm_j}_4^4}.
\]
\end{corollary}
\begin{proof}
Using Cauchy-Schwarz twice, 
\begin{multline*}
\sum_{j=1}^n \rnorm{\bThm_j}_3^3 = \sum_{j=1}^n \E\left[\len{\bThm_j}^3\right] = \sum_{j=1}^n \E\left[\len{\bThm_j} \len{\bThm_j}^2\right] \leq \sum_{j=1}^n \sqrt{\E\left[\len{\bThm_j}^2\right]}\sqrt{\E\left[\len{\bThm_j}^4\right]} \\\leq \sqrt{\sum_{j=1}^n \E\left[\len{\bThm_j}^2\right]}\sqrt{\sum_{j=1}^n \E\left[\len{\bThm_j}^4\right]} = \sqrt{d}\sqrt{\sum_{j=1}^n \rnorm{\bThm_j}_4^4},
\end{multline*}
where we also used $\sum \sigma_j^2 = d$.
\end{proof}

\subsection{Derandomization and hypercontractivity}
\label{sec:derand}
We now show that this result can be ``derandomized'' in a certain
sense.  This idea is essentially due to Meka and
Zuckerman~\cite[Sec.~4.1]{MZ09}. 

\begin{definition} We say that the sequences of $\R^d$-valued random
  vectors $\bThm_1, \dots, \bThm_n$ and $\bDerand_1, \dots,
  \bDerand_n$ satisfy the \emph{$r$-matching-moments condition}, $r
  \in \N$, if the following holds: $\E[\boldsymbol{\mathcal{X}}^K] =
  \E[\boldsymbol{\mathcal{Y}}^K]$ for all multi-indices $\abs{K} \leq
  r$, where $\boldsymbol{\mathcal{X}}$ is the $\R^{dn}$-valued random
  vector gotten by concatenating $\bThm_1, \dots, \bThm_n$, and
  $\boldsymbol{\mathcal{Y}}$ is defined similarly. 
\end{definition}

In this section, we suppose that $\bDerand_1, \dots, \bDerand_n$ satisfy the $4$-matching-moments condition with respect to $\bThm_1, \dots, \bThm_n$.  We will \emph{not} suppose that they are independent, but rather that they have some limited independence.  Let $\bT = \bDerand_1 + \cdots + \bDerand_n$.
\begin{proposition}  Let $H_1, \dots, H_t$ form a partition of $[n]$, and write $\bZ_\ell = \sum_{j \in H_\ell} \bDerand_j$.  Assume that $\bZ_1, \dots, \bZ_t$ are independent.  Then
\[
\abs{\E[\psi(\bT)] - \E[\psi(\bG)]} \leq 2\nice d^{7/2} \sqrt{\sum_{\ell=1}^t \rnorm{\littlesum_{j \in H_\ell} \bThm_j}_4^4}.
\]
\end{proposition}
\begin{proof}
We simply apply Corollary~\ref{cor:fourth} to the random variables $\bZ_1, \dots, \bZ_t$.  To check that it is applicable, we note the following:  The random variables are independent. They satisfy $\E[\bZ_\ell] = 0$, because each $\E[\bDerand_j] = 0$ by $1$-matching-moments.  The covariance matrix $\sum_{\ell=1}^t \Cov[\bZ_\ell] = \cmatrix$, by $2$-matching-moments. 

Thus Corollary~\ref{cor:fourth} gives
\[
\abs{\E[\psi(\bT)] - \E[\psi(\bG)]} \leq 2\nice d^{7/2} \sqrt{\littlesum_{\ell=1}^t \rnorm{\bZ_\ell}_4^4}.
\]
But for each $\ell$,
\[
\rnorm{\bZ_\ell}_4^4 = \rnorm{\littlesum_{j \in H_\ell} \bDerand_j}_4^4 = \E\Bigl[\langle \littlesum_{j \in H_\ell} \bDerand_j, \littlesum_{j \in H_\ell} \bDerand_j \rangle^2\Bigr]
= \E\Bigl[\langle \littlesum_{j \in H_\ell} \bThm_j, \littlesum_{j \in H_\ell} \bThm_j \rangle^2\Bigr] = \rnorm{\littlesum_{j \in H_\ell} \bThm_j}_4^4,
\]
using $4$-matching-moments, completing the proof.
\end{proof}
\begin{remark} The full $4$-matching-moments condition is not essential for our results; it would suffice to have $2$-matching-moments, along with a good upper bound on the $4$th moments of the $\bDerand_j$'s with respect to those of the $\bThm_j$'s.
\end{remark}

We can simplify the previous bounds if we assume hypercontractivity.
\begin{corollary}  \label{cor:derand-hc} If we additionally assume that the random vectors $\bThm_1, \dots, \bThm_n$ are $\eta$-HC, then we have
\begin{eqnarray*}
\abs{\E[\psi(\bS)] - \E[\psi(\bG)]} &\leq& (2\nice d^{7/2}/\eta^2) \sqrt{\littlesum_{j=1}^n \sigma_j^4},\\
\abs{\E[\psi(\bT)] - \E[\psi(\bG)]} &\leq& (2\nice d^{7/2}/\eta^2) \sqrt{\littlesum_{\ell=1}^t \Bigl(\littlesum_{j \in H_\ell} \sigma_j^2 \Bigr)^2}.
\end{eqnarray*}
\end{corollary}
\begin{proof}  We prove only the second statement, the first being simpler.  It suffices to show
\[
\rnorm{\littlesum_{j \in H_\ell} \bThm_j}_4^4 \leq (1/\eta)^4 \Bigl(\littlesum_{j \in H_\ell} \sigma_j^2\Bigr)^2.
\]
Since the random variables $\{\bThm_j : j \in H_\ell\}$ are independent and $\eta$-HC, it follows that the (vector-valued) random variable $\littlesum_{j \in H_\ell} \bThm_j$ is $\eta$-HC.  Hence
\[
\rnorm{\littlesum_{j \in H_\ell} \bThm_j}_4^4 \leq (1/\eta)^4 \Bigl(\rnorm{\littlesum_{j \in H_\ell} \bThm_j}_2^2\Bigr)^2.
\]
But
\[
\rnorm{\littlesum_{j \in H_\ell} \bThm_j}_2^2 = \littlesum_{j \in H_\ell} \sigma_j^2
\]
by the Pythagorean Theorem.
\end{proof}

We now consider the case when the partition $\bH_1, \dots, \bH_t$
chosen randomly using a $\collision$-collison preserving family of hash
functions (see Definition \ref{def:hash-collision}).

\begin{proposition} \label{prop:derand} In the setting of
  Corollary~\ref{cor:derand-hc}, if the partition $\bH_1, \dots,
  \bH_t$ is chosen using a $\collision$-collision preserving family of hash functions, then
\[
\abs{\E[\psi(\bT)] - \E[\psi(\bG)]} \leq (2\nice \collision^{1/2} d^{7/2}/\eta^2) \sqrt{\frac{d^2}{t} + \littlesum_{j=1}^n \sigma_j^4}.
\]
where the expectation $\E[\psi(\bT)]$ is with respect to both the choice of $\bH_1, \dots, \bH_t$ and $\bDerand_1, \dots, \bDerand_n$.
\end{proposition}
\begin{proof}
By the triangle inequality for real numbers, it suffices to show
\[
\E_{\bH_1, \dots, \bH_t}\left[\sqrt{\littlesum_{\ell=1}^t \Bigl(\littlesum_{j \in \bH_\ell} \sigma_j^2\Bigr)^2}\right] \leq \sqrt{\collision \left(\frac{d^2}{t} + \littlesum_{j=1}^n \sigma_j^4 \right)}.
\]
By Cauchy-Schwarz, this reduces to showing
\[
\E_{\bH_1, \dots, \bH_t}\left[\littlesum_{\ell=1}^t \Bigl(\littlesum_{j \in \bH_\ell} \sigma_j^2\Bigr)^2\right] \leq \collision \left(\frac{d^2}{t} + \sum_{j=1}^n \sigma_j^4 \right).
\]
But
\begin{multline*}
\E_{\bH_1, \dots, \bH_t}\left[\littlesum_{\ell=1}^t \Bigl(\littlesum_{j \in \bH_\ell} \sigma_j^2\Bigr)^2\right] = \sum_{\ell=1}^t \E\left[\Bigl(\littlesum_{j=1}^n {\bf 1}_{\{j \in \bH_\ell\}}\sigma_j^2\Bigr)^2\right] = \sum_{\ell=1}^t \sum_{j_1,j_2=1}^n \sigma_{j_1}^2\sigma_{j_2}^2\E[{\bf 1}_{\{j_1 \in \bH_\ell\}}{\bf 1}_{\{j_2 \in \bH_\ell\}}]
\\ \leq \sum_{\ell=1}^t \left(\frac{\collision}{t}\sum_{j=1}^n \sigma_j^4 \right) + \sum_{j_1 \neq j_2} \sigma_{j_1}^2 \sigma_{j_2}^2 \sum_{\ell=1}^t \E[{\bf 1}_{\{j_1 \in \bH_\ell\}}{\bf 1}_{\{j_2 \in \bH_\ell\}}] \leq \collision\sum_{j=1}^n \sigma_j^4 + \frac{\collision}{t}\sum_{j_1 \neq j_2} \sigma_{j_1}^2\sigma_{j_2}^2 \leq \frac{\collision d^2}{t} + \collision \sum_{j=1}^n \sigma_j^4,
\end{multline*}
as needed, because
\[
\sum_{j_1 \neq j_2} \sigma_{j_1}^2\sigma_{j_2}^2 \leq \left(\sum_{j=1}^n \sigma_j^2\right)^2 = d^2.
\]
\end{proof}

\subsection{Smoothing}

Ideally we would like to use the results from the previous sections with $\psi$ equal to certain indicator functions $\chi : \R^d \to \{0,1\}$; however these are not $\calC^3$.  As usual in the Lindeberg method (see, e.g.,~\cite{MOO05}), we overcome this by working with mollified versions of these functions.  For most of this section, we will work with our underandomized result, the statement about~$\bS$ in Corollary~\ref{cor:fourth}.  Identical considerations apply to the statement about $\bT$ in Proposition~\ref{prop:derand}, and we will draw the necessary conclusions at the end.\\

Let $\xi : \R \to \R$ be the ``standard mollifier'', a smooth density function supported on $[-1,1]$.  We will use the fact that there is some universal constant $\nice_0$ such that $\int \abs{\xi^{(k)}} dx \leq \nice_0$ for $k = 1, 2, 3$ (where $\xi^{(k)}$ denotes the $k$th derivative of $\xi$).  Given $\eps > 0$ we define $\xi_{\epsilon}(x) = \xi(x/\epsilon)/\epsilon$, the standard mollifier with support $[-\epsilon, \epsilon]$.  Finally, define the density function $\Xi_\epsilon$ on $\R^d$ by $\Xi_\epsilon(x_1, \dots, x_d) = \prod_{i=1}^d \xi_\epsilon(x_i)$.  We now prove an elementary lemma:
\begin{lemma} \label{lem:smoothing}
Let $\chi : \R^d \to [-1,1]$ be measurable, let $\epsilon > 0$, and define $\psi = \Xi_\epsilon \ast \chi$, a smooth function.  Then for any multi-index $\abs{K} = 3$ we have $\abs{\psi^{(K)}} \leq (\nice_0/\epsilon)^3$.\rnote{check}
\dnote{did we define $\ast$?}
\end{lemma}
\begin{proof}
Using the fact that $\abs{\chi} \leq 1$ everywhere, we have
\[
\abs{\psi^{(K)}(a)} = \abs{\Xi_\epsilon^{(K)} \ast \chi(a)} \leq \int \abs{\Xi_\epsilon^{(K)}} = \int_{[-\epsilon,\epsilon]^d} \abs{\prod_{i=1}^d\frac{\partial^{k_i}}{\partial x_{i}^{k_i}} \xi_\epsilon(x_i)}\,dx_1\cdots dx_d = \prod_{i=1}^d \int_{-\epsilon}^\epsilon  \abs{\frac{\partial^{k_i}}{\partial x^{k_i}} \xi_\epsilon(x)}\,dx.
\]
Note that $\frac{\partial^k}{\partial x^k} \xi_\epsilon(x) = \xi^{(k)}(x/\epsilon)/\epsilon^{k+1}$, from which it follows that 
\[
\int_{-\epsilon}^\epsilon  \abs{\frac{\partial^{k}}{\partial x^{k}} \xi_\epsilon(x)}\,dx \leq \nice_0/\epsilon^{k}
\]
for $k = 1, 2, 3$.  For $k = 0$ we of course have 
\[
\int_{-\epsilon}^\epsilon  \abs{\xi_\epsilon(x)}\,dx = \int_{-\epsilon}^\epsilon  \xi_\epsilon(x) \,dx = 1.
\]
Since $\abs{K} = 3$, we therefore achieve the claimed upper bound of $(\nice_0/\eps)^3$.
\end{proof}

Suppose now $A \subseteq \R^d$ is a measurable set.  We define:
\[
A^{+\epsilon} = \{x \in \R^d : x + [-\epsilon/2, \eps/2]^d \cap A \neq \emptyset\}, \quad A^{-\epsilon} = \{x \in \R^d : x + [-\epsilon/2, \eps/2]^d \subseteq A\}, \quad \Game^{\epsilon} A = A^{+\epsilon} \setminus A^{-\epsilon}.
\]
We also define $\psi_{A^{+\epsilon}} = \Xi_\epsilon \ast \chi_{A^{+\epsilon}}$ as in Lemma~\ref{lem:smoothing}, where $\chi_{A^{+\epsilon}}$ is the $0$-$1$ indicator of $A^{+\epsilon}$, and similarly define $\psi_{A^{-\epsilon}}$.  Applying now Corollary~\ref{cor:fourth}, we conclude:
\begin{lemma} \label{lem:cor} For $\psi = \psi_{A^{+\epsilon}}$ or $\psi = \psi_{A^{-\epsilon}}$ it holds that
\[
\abs{\E[\psi(\bS)] - \E[\psi(\bG)]} \leq (2\nice_0d^{7/2}/\eta^2\eps^3) \sqrt{\littlesum_{j=1}^n \sigma_j^4}.
\]
\end{lemma}

It is clear from the definitions that both $\psi_{A^{+\epsilon}}$ and $\psi_{A^{-\epsilon}}$ have range $[0,1]$, and that pointwise, $\psi_{A^{-\epsilon}} \leq \chi_A \leq \psi_{A^{+\epsilon}}$.
Thus
\[
\E[\psi_{A^{-\epsilon}}(\bS)] \leq \Pr[\bS \in A] \leq \E[\psi_{A^{+\epsilon}}(\bS)],
\]
\[
\E[\psi_{A^{-\epsilon}}(\bG)] \leq \Pr[\bG \in A] \leq \E[\psi_{A^{+\epsilon}}(\bG)].
\]
From Lemma~\ref{lem:cor} we have that the two left-hand sides above are close and that the two right-hand sides are close.  Because of good anti-concentration of Gaussians, it may also be that the left-hand and right-hand sides on the second line are also close, in which $\Pr[\bS \in A]$ and $\Pr[\bG \in A]$ will also be close.  This motivates the following observation: $\psi_{A^{+\epsilon}} = \psi_{A^{-\epsilon}} = 1$ on $A^{-\epsilon}$ and $\psi_{A^{+\epsilon}} = \psi_{A^{-\epsilon}} = 0$ on the complement of $A^{+\epsilon}$.  Hence
\[
\E[\psi_{A^{+\epsilon}}(\bG)] - \E[\psi_{A^{-\epsilon}}(\bG)] \leq \Pr[\bG \in \Game^{\epsilon} A].
\]
Putting together these observations, we conclude:
\begin{theorem} \label{thm:close} We have
\[
\abs{\Pr[\bS \in A] - \Pr[\bG \in A]} \leq (2\nice_0d^{7/2}/\eta^2\eps^3) \sqrt{\littlesum_{j=1}^n \sigma_j^4} + \Pr[\bG \in \Game^{\epsilon} A].
\]
\end{theorem}

\subsection{Translates of unions of orthants}
Let us now specialize to the case where $A \subseteq \R^d$ is a
translate of a union of orthants.  Recall that this means that there exists some
vector $\Theta \in \R^d$ such that $X \in A$ depends only on $\vsgn(X
- \Theta)$. \ignore{(The definition of $\sgn(0)$ will not matter here;
  i.e., it does not matter which parts of its topological boundary $A$
  includes.)}  
We make the following observation, whose proof is trivial.

\begin{proposition} \label{prop:orths} If $A \subseteq \R^d$ is a union of orthants then
\[
\Game^{\epsilon} A \subseteq \bigcup_{i=1}^d W_i^\eps,
\]
where
\[
W_i^\eps = \{X \in \R^d : \abs{X[j] - \Theta[j]} \leq \eps/2\}.
\]
\end{proposition}

But we also have the following:
\begin{proposition} \label{prop:walls} Assuming the $d$-dimensional Gaussian $\bG$ with covariance matrix $\cmatrix$ satisfies $\cmatrix[i,i] = 1$ for all $i \in [d]$, it holds that
\[
\Pr\left[\bG \in \bigcup_{i=1}^d W_i^\eps\right] \leq d\epsilon/\sqrt{2\pi}.
\]
\end{proposition}
\begin{proof} By a union bound it suffices to prove that $\Pr[\abs{\bG[i] - \Theta[i]} \leq \eps/2] \leq \eps/\sqrt{2\pi}$.  This is straightforward, as $\bG[i]$ has distribution $N(0,1)$ and hence has pdf bounded above by $1/\sqrt{2/\pi}$.
\end{proof}

We  now prove Theorem \ref{thm:be}

\begin{proof} {\bf (Theorem \ref{thm:be})} 
For any $\eps > 0$, we may combine Propositions~\ref{prop:orths} and~\ref{prop:walls} with Theorem~\ref{thm:close} and conclude
\[
\abs{\Pr[\bS \in A] - \Pr[\bG \in A]} \leq (2\nice_0d^2/\eta^2\eps^3) \sqrt{\littlesum_{j=1}^n \sigma_j^4} + d\epsilon/\sqrt{2\pi}.
\]
The proof is completed by taking $\epsilon = \eta^{-1/2}d^{5/8}(\littlesum_{j=1}^n \sigma_j^4)^{1/8}$ (which is strictly positive since $\sum \sigma_j^4 = 0$ is impossible).
\end{proof}

Identical reasoning gives the proof of Theorem
\ref{thm:be-derandomized}. Combining Theorems \ref{thm:be} and
\ref{thm:be-derandomized} gives Theorem \ref{thm:very-long}.

\ignore{
We conclude:
\begin{theorem}  
\label{thm:very-long}
Let $\bThm_1, \dots, \bThm_n$ be independent $\R^d$-valued random
vectors which are $\eta$-HC.  Write $\sigma_j^2 = \rnorm{\bThm_j}_2^2$
and assume that $\sum_{j=1}^n \E[\bThm_j[i]^2] = 1$ for each $i \in
[d]$.  Write $\bS = \sum_{j=1}^n \bThm_j$.  Next, let $\bDerand_1,
\dots, \bDerand_n$ be $\R^d$-valued random vectors which satisfy the
$4$-matching-moments condition with respect to $\bThm_1, \dots,
\bThm_n$.  Further, assume that generation of $\bDerand_1, \dots,
\bDerand_n$ satisfies the following property: First, a
pairwise-independent random partition $\bH_1, \dots, \bH_t$ of $[n]$
is chosen; then, conditioned on this, the random vectors $\bZ_\ell =
\sum_{j \in \bH_\ell} \bDerand_j$ are independent.  Write $\bT =
\bDerand_1 + \cdots + \bDerand_n$.  Finally, assume that $A \subseteq
\R^d$ is a translated union of orthants; i.e., there exists $\Theta
\in \R^d$ such that $X \in A$ depends only on $\vsgn(X - \Theta)$.}

\section{Critical Index for Hypercontractive Random Variables}
\label{sec:crit}

In this section, we generalize the critical index to random variables that are hypercontractive. We will consider $\eta$-HC random variables $\bx_0, \dots, \bx_n$ which are at least pairwise independent.  Write $\sigma_j^2 = \rnorm{\bx_j}_2^2$, and note that pairwise independence implies
$\rnorm{\bx_0 + \cdots + \bx_n}_2^2 = \sigma_0^2 + \dots +
\sigma_n^2$. We also write $\tau_i^2 = \rnorm{\bx_{i} + \bx_{i+1} +
  \cdots + \bx_{n}}_2^2 = \sum_{j \geq i} \sigma_j^2$. 

\begin{definition}
\label{def:reg}
For $0 < \delta < 1$, we say that the collection of random variables
$\bx_0, \dots, \bx_n$ is  \emph{$\delta$-regular} if $\sum_{j=0}^n
\rnorm{\bx_j}_4^4 \leq \delta \Bigl(\sum_{j=0}^n
\rnorm{\bx_j}_2^2\Bigr)^2 = \delta \tau_0^4$. 
\end{definition}

\begin{definition}  Suppose the sequence $\bx_0, \dots, \bx_n$ is \emph{ordered}, meaning that 
$\sigma_0^2 \geq \sigma_1^2 \geq \sigma_2^2 \geq \cdots$. Then for $0
  < \delta < 1$, the \emph{$\delta$-critical index} is defined to be
  the smallest index~$\ell$ such that the sequence $\bx_\ell,
  \bx_{\ell+1}, \dots, \bx_{n}$ is $\delta$-regular, or $\ell =
  \infty$ no such index exists. 
\end{definition}

\begin{theorem}  
\label{thm:crit-one} 
Let $0 < \delta < 1$, $0 < \eps < 1/2$, and $s > 1$ be parameters.
Let $L = br$, where $b = \lceil (2/\eta^4)\ln(1/\eps) \rceil$ and $r =
\lceil (1/\eta^4\delta)\ln(1+16s^2)\rceil$; note that 
\[ L \leq O\left(\frac{\log(s) \log(1/\eps)}{\eta^8}\right) \cdot \frac{1}{\delta}.\]
Assume the sequence $\bx_0, \dots, \bx_n$ is ordered, that $n \geq L$,
and that $\bx_0, \dots, \bx_{L-1}$ are independent.  Then if $\ell$ is
the $\delta$-critical index for the sequence, and $\ell \geq L$, then
for all $\theta \in \R$, 
\[
\Pr\left[\left|\bx_0 + \cdots + \bx_{L-1} -\theta\right| \leq s \cdot \tau_{L} \right] \leq \eps + \frac{O(\ln(1/\eps))}{\eta^8 s^4}.
\]
\end{theorem}
\begin{proof}
For any $0 \leq j < L$, since the critical index $\ell$ is at least $j$ we have
\begin{align*}
\delta \tau_j^4 <  \sum_{i \geq j} \rnorm{\bx_i}_4^4  \leq (1/\eta^4)
\sum_{i \geq j} \sigma_i^4 &\text{(since each $\bx_i$ is $\eta$-HC)}  \leq (\sigma_j^2/\eta^4) \sum_{i \geq j} \sigma_i^2  = (\sigma_j^2/\eta^4) \tau_j^2. &
\end{align*}
where we used hypercontractivity and the fact that $\sigma_i$s are
ordered. Hence for all $0 \leq j < L$, 
\[
\eta^4 \delta \tau_j^2 < \sigma_j^2 = \tau_j^2 - \tau_{j+1}^2 \quad \Rightarrow \quad \tau_{j+1}^2 < (1-\eta^4 \delta) \tau_j^2.
\]
It follows that for all $0 \leq k < b$,
\begin{equation} \label{eqn:tau-decay}
\tau_{(k+1)r}^2 < (1-\eta^4 \delta)^{r} \tau_{kr}^2 < \frac{1}{1+16s^2} \tau_{kr}^2,
\end{equation}
where we used the definition of $r$.  

Now for each $0 \leq k < b$ define $\by_k = \bx_{kr} + \bx_{kr+1} +
\bx_{kr+2} + \cdots + \bx_{(k+1)r-1}$ and
$\upsilon_k^2 = \rnorm{\by_k}_2^2 = \tau_{kr}^2 - \tau_{(k+1)r}^2$.
Using~\eqref{eqn:tau-decay} we have immediately conclude
\begin{equation} \label{eqn:upsilon-vs-tau}
\upsilon_k^2 > 16s^2\tau_{(k+1)r}^2 \qquad \Rightarrow \qquad \upsilon_k > 4s \tau_{(k+1)r}.
\end{equation}
Since all of $\bx_0, \dots, \bx_{L-1}$ are independent and $\eta$-HC,
we have that $\by_0, \by_1, \dots, \by_{b-1}$ are independent
$\eta$-HC random variables.  For $0 \leq k < b$, define the event $
A_k = \text{``}|\by_0 + \by_1 + \cdots + \by_{k} - \theta| \leq (1/2) \upsilon_k\text{'',}$.
We claim that for any $0 \leq k < b$,
\[
\Pr[A_k \mid A_{0} \wedge A_1 \wedge \cdots \wedge A_{k-1}] < 1 - \eta^4/2.
\]
To see this, note that conditioning only affects the values of random variables $\by_0, \dots, \by_{k-1}$, of  which $\by_k$ is independent.  Further, for \emph{every} choice of values for $\by_0, \dots, \by_{k-1}$, the event $A_k$ is an anti-concentration event of the type in Proposition~\ref{prop:HC-anticonc}, with some shifted $\theta$.  Hence the claim follows from this Proposition, as $(1-(1/2)^2)^2 > 1/2$.  Having established the claim, we conclude
\begin{equation} \label{eqn:and_A}
\Pr[A_0 \wedge A_1 \wedge \cdots \wedge A_{b-1}] < (1- \eta^4/2)^b \leq \eps.
\end{equation}

Let us now define, for each $1 \leq k < b$, random variables $
\bz_{k} = \by_{k} + \by_{k+1} + \cdots + \by_{b-1}$.
These random variables are also $\eta$-HC, and they satisfy
$\rnorm{\bz_{k}}_2^2 \leq \tau_{kr}^2$.
If we define the events $
B_{k} = \text{``}|\bz_{k}| \geq s \tau_{kr}\text{'',}$
then Proposition~\ref{prop:HC-conc} implies $\Pr[B_k] \leq 1/\eta^4s^4$.  Hence
\begin{equation} \label{eqn:or_B}
\Pr[B_1 \vee B_2 \vee \cdots \vee B_{b-1}] \leq (b-1)/\eta^4s^4 < b/\eta^4s^4.
\end{equation}

Combining~\eqref{eqn:and_A} and~\eqref{eqn:or_B} we see that except with probability less than 
$\eps + b/\eta^4s^4 \leq \eps + \frac{O(\ln(1/\eps))}{\eta^8 s^4}$,
at least one event $\overline{A_{k}}$ occurs, and none of the events $B_k$ occurs.  Since this is the error bound in the Theorem, it remains to show that in this case, the desired result ``$|\bx_0 + \cdots + \bx_{L-1} -\theta| > s \cdot \tau_{L}$'' occurs. Assume then that $\overline{A_m}$ occurs and $B_{m+1}$ does not occur, $0 \leq m < b$.  (For $m = b-1$ we need not make the latter assumption.)  Thus
\[
|\by_0 + \by_1 + \cdots + \by_{m} - \theta| > (1/2) \upsilon_m \quad \text{and} \quad |\bz_{m+1}| \leq s \tau_{(m+1)r} < (1/4) \upsilon_{m},
\]
where we used~\eqref{eqn:upsilon-vs-tau}.  (This makes sense also in the case $m = b-1$ if we naturally define $\bz_{b} \equiv 0$.)  By definition of $\bz_{m+1}$, we therefore obtain
\[
|\by_0 + \by_1 + \cdots + \by_{b-1} - \theta| = \left|\bx_0 + \cdots + \bx_{L-1} -\theta\right| > (1/4) \upsilon_m \geq (1/4) \upsilon_{b-1} \geq s \tau_{br} = s \tau_L,
\]
as desired, where we used~\eqref{eqn:upsilon-vs-tau}.
\end{proof}

We now state the high-dimensional generalization of Theorem
\ref{thm:crit-one}.
Assume $\bx_1, \dots, \bx_n$ are $\eta$-HC real random variables which
are at least pairwise independent. Assume also that $W_1, \dots, W_n$ are arbitrary fixed
vectors in $\R^d$, and write $\bX_j = \bx_j W_j$. 
\begin{theorem}
\label{thm:crit-many}
Let $\delta, \eps, s, L$ be as in Theorem~\ref{thm:crit-one}.  Then
there exists a set of coordinates $H_0 \subseteq [n]$, $\abs{H_0} \leq
dL$, with the following property.  Assuming the collection of random
variables $\{\bx_j : j \in H_0\}$ is independent, for each coordinate
$i \in [d]$ we have either: 
\begin{tightenumerate} 
\item the sequence of real random variables $\{\bX_j[i] : j \not \in H_0\}$ is $\delta$-regular; or,
\item for all $\theta \in \R$,
\[
\Pr\left[\Bigl|\sum_{j \in H_0} \bX_j[i] -\theta\Bigr| \leq s \cdot \sqrt{\sum_{j \not \in H_0} \rnorm{\bX_j[i]}_2^2}\right] \leq \eps + \frac{O(\ln(1/\eps))}{\eta^8 s^4}.
\]
\end{tightenumerate}
\end{theorem}

The fact that the sequence $\bx_0, \dots, \bx_n$ was ordered by
decreasing $2$-norm in Theorem~\ref{thm:crit-one} was mainly used for
notational convenience.  We can extract from the proof the following
corollary for unordered sequences (whose proof we omit):

\begin{corollary} \label{cor:crit-one} Let $\delta, \eps, s, b, r, L$ be as in Theorem~\ref{thm:crit-one}.  For the unordered collection $\bx_0, \dots, \bx_n$, assume we have a sequence of indices $0 \leq j_0 < j_1 < \cdots < j_{L-1} < n$ such that:
\begin{itemize}
\item for each $0 \leq t < L$, $\sigma_{j_t}^2 \geq \sigma_{j'}^2$ for all $j' > j_t$;
\item for each $0 \leq t < L$, $\{\bx_{j_t}, \bx_{j_t + 1}, \dots, \bx_{n}\}$ is \emph{not} $\delta$-regular.
\end{itemize}
Assume also that $\bx_0, \dots, \bx_{j_L}$ are independent.  Then for all $\theta \in \R$,
\[
\Pr\left[\left|\bx_0 + \cdots + \bx_{j_{L-1}} -\theta\right| \leq s \cdot \tau_{j_{L-1}+1} \right] \leq \eps + \frac{O(\ln(1/\eps))}{\eta^8 s^4}.
\]
\end{corollary}
The case when $j_t = t$ for $0 \leq t < L$ corresponds to Theorem~\ref{thm:crit-one}.

\ignore{
\begin{proof}  We follow the proof of Theorem~\ref{thm:crit-one}.  In place of its first deduction, we use that for each $0 \leq t < L$,
\[
\delta \tau_{j_{t}}^4 < \sum_{i \geq j_t} \rnorm{\bx_i}_4^4 \leq (1/\eta^4) \sum_{i \geq j_t} \sigma_i^4 \leq (\sigma_{j_t}^2/\eta^4) \sum_{i \geq j_2} \sigma_i^2 = (\sigma_{j_t}^2/\eta^4) \tau_{j_t}^2.
\]
Here we used both properties of the sequence $j_0, \dots, j_{L-1}$.  Thus (defining also $j_{L} = j_{L-1}+1$),
\[
\eta^4 \delta \tau_{j_t}^2 < \sigma_{j_t}^2 \leq \tau_{j_t}^2 - \tau_{j_{t+1}}^2 \quad\Rightarrow\quad \tau_{j_{t+1}}^2 < (1-\eta^4 \delta)\tau_{j_t}^2.
\]
It follows that we can extract a subsequence of the $j_t$'s of length $b$ --- call it $j'_0, j'_1, \dots, j'_{b-1}$ --- such that
\[
\tau_{j'_{k+1}}^2 < \frac{1}{1+16s^2}\tau_{j'_k}^2.
\]
Once we define
\[
\by_k = \bx_{j'_k} + \bx_{j'_k + 1} + \cdots + \bx_{j'_{k+1}-1},
\]
the remainder of the proof is essentially the same.  The only additional trick is to handle the random variables $\bx_0, \dots, \bx_{j'_0-1}$ by noting that the anti-concentration holds for \emph{every} realization of these variables, by modifying $\theta$.\rnote{this whole proof is poorly explained}
\end{proof}}

We now prove Theorem \ref{thm:crit-many}.

\begin{proof}
We construct $H_0$ according to an iterative process. Initially, $H_0
= \emptyset$, and we define $c_i = 0$ for all $i \in [d]$.  In each
step of the process,  we do the following: First, we select any $i$
such that $c_i < L$ and such that the collection $\{\bX_j[i] : j \not
\in H_0\}$ is \emph{not} $\delta$-regular.  If there is no such $i$
then we stop the whole process.  Otherwise, we continue the step by
choosing $j \in [n] \setminus H_0$ so as to maximize
$\rnorm{\bX_j[i]}_2^2$.  We then end the step by adding $j$ into $H_0$
and incrementing $c_i$. 

Note that the process must terminate with $\abs{H_0} \leq dL$; this is
because each step increments one of $c_1, \dots, c_d$, but no $c_i$
can exceed $L$. When the process terminates, for each $i$ we have
either that $\{\bX_j[i] : j \not \in H_0\}$ is $\delta$-regular or
that $c_i = L$. 

It suffices then to show that when $c_i = L$, the anti-concentration
statement holds for $i$.  To see this, first reorder the sequence of
random variables $(\bX_j[i])_j$ so that the first $\abs{H_0}$ are in
the order that the indices were added to $H_0$, and the remaining $n -
\abs{H_0}$ are in an arbitrary order.  Write $1 \leq j_0 < j_1 <
\cdots < j_{L-1} \leq \abs{H_0}$ for the indices that were added to
$H_0$ on those steps which incremented $c_i$.  Then by the definition
of the iterative process, for each $0 \leq t < L$ we have that
$\rnorm{\bX_{j_t}[i]}_2^2 \geq \rnorm{\bX_{j'}[i]}_2^2$ for all $j' >
j_t$ and that $\{\bX_{j_t}[i], \bX_{j_t+1}[i], \cdots, \bX_{n}[i]\}$
is not $\delta$-regular. The anti-concentration statement now follows
from Corollary~\ref{cor:crit-one}. 
\end{proof}

\section{The Meka-Zuckerman Generator}
\label{ap:mzg}

\newcommand{\greg}{G}
\newcommand{\gzs}{r_0}
\newcommand{\setsize}{\ell}

\newcommand{\fr}[1]{\frac{1}{#1}}
\newcommand{\sd}{\mathrm{SD}}

For the Meka-Zuckerman generator, the first step is to reduce the problem of fooling
functions of halfspaces under an arbitrary $C$-bounded product distribution to fooling 
an $O(C)$-bounded {\em discrete} product distribution with
support $\poly(n,C, \eps^{-1})$ in each co-ordinate.

\begin{lemma}
\label{lem:d-halfspaces}
Given a $C$-bounded distribution $\bX$, there is a discrete product
distribution $\bY$ such that if $f: \mb{R}^n \rightarrow \pmo$ is a
function of $d$ halfspaces $\{h_i:\mb{R}^n \rightarrow \pmo\}_{i \in [d]}$,
then
$$|\E[f(\bX)] - \E[f(\bY)]| \leq O\left(\frac{d\eps^2}{nC^2}\right).$$ 
Each $\by_i$ is distributed uniformly over a multiset $\domain_i = \{b_1(i) \leq
\cdots \leq b_g(i)\}$ where $|b_j(i)| \leq (nC^2\eps^{-1})^\frac{1}{4}$. 
For every $i$, we have $|\Omega_i| = 2^s = O(n^2C^2\eps^{-2})$ and
Further $\E[\by_i] = 0, \E[\by^2_i] =1, \ \E[\by^4_i]  \leq O(C)$. 
\end{lemma}

We are interested in $d << n$, so the error in going from $\bX$ to $\bY$ is
$o(\eps^2)$. Since $|\domain_i|  =2^s$ for all $i$, sampling $k$-wise independently from $\bY$
reduces to generating $n$ strings of length $s$ in a
$k$-wise independent manner: this can be done using $k\max(\log n,s) =
O(k\log(nC/\eps))$ random bits.

This lemma is proved by {\em sandwiching} $\bX$ between two discrete product distributions $\bY^u$ and
$\bY^\ell$ which are close to each other in statistical distance. The
proof is in section~\ref{discretize}. Henceforth, we will rename
$\bY$ as $\bX$ and focus on fooling discrete product distributions.

We now describe the main generator of Meka-Zuckerman, modified so that
random variables take values in $\prod_j\domain_j$ instead of simply $\pm
1$. At a high level, the generator hashes variables into buckets and uses bounded
independence for the variables within each bucket. We use a weaker
property of hash functions than used in \cite{MZ09}.

The generator first picks a partition of $[n] = H_1 \cup \ldots \cup H_t$ 
using a random element from~$\calH$, a 1-collision preserving family 
of hash functions.
For each $i \in [t]$, it then generates a 5-wise independent distribution $(\by_j)_{j \in H_i}$ on
$\prod_{j\in H_i}\domain_j$. Such a distribution on $n$ random
variables can be generated using a seed of length $k \log
\max(n,|\Omega|)$. These $t$ distributions are chosen independently. The generator
outputs $\bY = (\by_1,\ldots,\by_n)$. The seedlength required is
$\log(2n) + 5t\log \max(n,|\Omega|)$ where $\log(2n)$ are required for
the hash function and $5\log\max(n,|\Omega|)$ bits are needed for each
$H_i$, $i\in [t]$.

\section{Analyzing the Meka-Zuckerman Generator}
\label{ap:mzg-analysis}

We first prove that the indices in the set $H_0$ are
likely to be hashed into distinct buckets.

\begin{definition}
\label{defhash-isolating}
A hash function $h:[n] \to [t]$ is 
\emph{$S$-isolating} if for all $x \neq y \in S$, $h(x) \neq h(y)$.
A family of hash functions $\calH = \set{h:[n] \to [t]}$ is 
$(\setsize,\beta)$-isolating if for any $S \subseteq [n]$, $|S| \leq \setsize$, 
\[
  \Pr_{h \in_u \calH}[\mbox{$h$ is not $S$-isolating}] \leq \beta .
\]
\end{definition}

A $\collision$-collision preserving hash family is likely to be isolating for
small sets:

\begin{lemma}
Assume $t$ is a power of 2. A $\collision$-collision preserving family of hash functions $\calH =
\set{h:[n] \to [t]}$ is $(\setsize,\beta)$-collision free for $\beta =
\collision\setsize^2/(2t)$.  
\end{lemma}

\begin{proof}
The expected number of collisions for a set $S$ is at most
\[ {|S| \choose 2} \frac {\collision}{t} \leq \frac{\collision\ell^2}{2t}.\]

By increasing $n$ to the next largest power of 2, since $t$ is a power of 2,
there is a field $\F$ of size $n$ where $t|n$.  Then there is a hash
family of size $n$ for $\collision=1$.  For any element $a \in \F$, define a hash function $h_a(x) = (ax) \bmod t$.  Here $x$ is viewed as a field element, the multiplication is done in the field, and
the product is then viewed as a nonnegative integer less than $n$ before taking the mod.
We can increase $n$ and $t$ to be the nearest powers of 2. We can therefore take $\calH$ to have size at most
$2n$ for $\collision=1$. 
\end{proof}

We want the set $H_0$ to be isolated with error $\eps$, so we want
$\setsize = dL$ and $\beta = \epsilon$. Hence we set $t$ to be the
smallest power of 2 larger than $\setsize^2/\epsilon = (dL)^2/\eps$.  

We will aim to achieve error $O(d\eps)$ (rather than $O(\eps)$), as
this makes  the notation easier. We set the parameters $s$ and
$\delta$ in Theorem \ref{thm:crit-many} as
\[ s = 1/(\eta^2 \sqrt \eps),
 \ \ \delta = \frac{\eta^4\eps^8}{d^{7}}.
 \]
 This implies that
 \[ L =O\left(\frac{\log(s) \log(1/\eps)}{\eta^8}\right)
 \cdot\frac{1}{\delta} =  O \left(\frac{d^{7}\log^2(\eps\eta)}{\eta^{12}\eps^8} \right),
\ \ t = O \left(\frac{(dL)^2}{\eps}\right) = O\left(\frac{d^{15}\log^4(\eps\eta)}{\eta^{24}\eps^{17}} \right). \]

\ignore{
We will need $s$ and $\delta$ to satisfy:
\begin{enumerate}
\item
$\frac{O(\log \eps^{-1})}{\eta^8 s^4} \leq \eps$.
\ignore{\dnote{Could even do $\leq d\eps$.}}
\item
$1/s^2 \leq \eps.$
\item
\label{long}
${\eta}^{-1/2} d^{15/8} \cdot (t^{-1} + \delta)^{1/8}
    \leq O(d\eps).$
\end{enumerate}

Since $t \geq L \geq 1/\delta$, condition~\ref{long} is equivalent to
\[ {\eta}^{-1/2} d^{7/8} \delta^{1/8}
    \leq O(\eps).
    \]

The following choices are sufficient:}

\subsection{Analysis for functions of regular halfspaces}

Recall that our goal is to fool functions of $\vsgn(\sum_jx_jW_j -
\theta)$. Let $\bY_j=\by_jW_j$ and $\bT = \sum_{j=1}^n \bY_j$. Similarly let $\bX_j = \bx_jW_j$
and $\bS = \sum_{j=1}^n \bX_j$. Thus we are interested in bounding
$$|\Pr_{\bX}[\bS \in A] - \Pr_{\bY}[\bT \in A]|$$
where $A$ is a translate of union of  orthants: membership  of a point $X \in
\mathbb{R}^d$ in $A$ is a function of $\vsgn(X -\Theta)$.  By rescaling the $W_j$ and $\Theta$, we may assume without loss of
generality that 
\[
\cmatrix[i,i] = \sum_{j=1}^n \E\left[\bThm_j[i]^2\right] = 1 \qquad \text{for all $i \in [d]$.}
\]
The regular case is when the vectors $W_1,\ldots,W_n$ are such that
for every $i$, the sequence of random variables $\{\bX_j[i]\}_{j=1}^n$
is $\delta$-regular. In this case, we can directly appeal to the
Berry-Esseen theorem to prove th correctness of the MZ generator.

\begin{theorem}
\label{thm:mz-regular}
If the sequence of random variables $\{\bX_j[i]\}_{j=1}^n$ is
$\delta$-regular for all $i \in [d]$, then the MZ generator
$O(d\eps)$-fools any function of $\sgn(W\cdot X - \Theta)$ for all
$\Theta \in \mb{R}^d$. 
\end{theorem}
\begin{proof}
We can therefore apply the machinery developed above.
For the regular case, we only need to use 4-wise independence.
Thus, the random variables $\bY_1, \dots, \bY_n$ satisfy the
$4$-matching-moments condition with respect to $\bX_1, \dots, \bX_n$,
as defined in Subsection~\ref{sec:derand}. 

The definition of $\delta$-regular is given in
Definition~\ref{def:reg}.
Let $\sigma_{i,j} = \rnorm{\bX_j[i]}_2$.
Suppose that for all~$i$, the set of real random variables $\set{\bX_j[i]}$ is $\delta$-regular, i.e.,
\[
\sum_{j=1}^n \rnorm{\bX_j[i]}_4^4 \leq \delta \Bigl(\sum_{j=1}^n \rnorm{\bX_j[i]}_2^2\Bigr)^2 = \delta (\sigma_{i,1}^2 + \dots + \sigma_{i,n}^2)^2 = \delta,
\]
where the last equality is from our normalization.
We wish to apply Theorem \ref{thm:very-long}.
Since $\sigma_{i,j} = \rnorm{\bX_j[i]}_2 \leq \rnorm{\bX_j[i]}_4$, we conclude that for all $i$,
\[
\sum_{j=1}^n \sigma_{i,j}^4 \leq \sum_{j=1}^n \rnorm{\bX_j[i]}_4^4 \leq \delta.
\]
Since $\sigma_j^2 = \sum_{i=1}^d \sigma_{i,j}^2$, by Cauchy-Schwarz we get
\[
\sigma_j^4 = \left(\sum_{i=1}^d \sigma_{i,j}^2 \right)^2 \leq d \left( \sum_{j=1}^n \sigma_{i,j}^4 \right) \leq d\delta.
\]
Therefore $\sum_{j=1}^n \sigma_j^4 \leq d^2 \delta$.  Hence we can apply Theorem~\ref{thm:very-long}
to obtain
\[
\abs{\Pr[\bS \in A] - \Pr[\bT \in A]} \leq O \left( (1/\eta)^{1/2}
d^{15/8}) \cdot (t^{-1} + \delta)^{1/8} \right) \leq O(d\eps).
\]
where the last inequality follows from the choice of $t, \delta$.
\end{proof}

\subsection{Analysis for functions of general halfspaces}
\newcommand{\reg}{\mathrm{REG}}
\newcommand{\junta}{\mathrm{JUNTA}}

We now combine Theorem \ref{thm:crit-many} with the analysis of the
Regular case (Theorem \ref{thm:mz-regular}), to prove that the MZ generator fools functions of
arbitrary halfspaces.

\begin{theorem}
\label{thm:mz-main}
The MZ generator $O(d\eps)$-fools any function of $d$ halfspaces with seed length
$$O(t\log(\max(n,|\Omega|))) = O \left(\frac{d^{15}\log^4(\eps\eta)\log(n/\eps\eta)}{\eta^{24}\eps^{17}} \right).$$
\end{theorem}

\begin{proof}
Apply Theorem \ref{thm:crit-many} with these parameters.  Then there 
exists a set $H_0 \subseteq [n]$ of size at most $dL$
such that the coordinates $[d]$ can be partitioned into two sets, $\reg$ and $\junta$, such that the
     following holds. 

\begin{enumerate} 
\item For $i \in \reg$, the set of real random variables $\{\bX_j[i] : j \not \in H_0\}$ is $\delta$-regular.
\item
\label{casejunta}
For $i \in \junta$, for all $\theta \in \R$,
\begin{equation}
\label{junta-bound}
\Pr\left[\Bigl|\sum_{j \in H_0} \bX_j[i] -\theta\Bigr| \leq s \cdot
  \sqrt{\sum_{j \not \in H_0} \rnorm{\bX_j[i]}_2^2}\right] \leq \eps +
\frac{O(\log(1/\eps))}{\eta^8 s^4} \leq 2\eps
\end{equation}
\end{enumerate}

\ignore{
We say that the hash function $\mathbf{h}$ is good if it is $\balance$-balanced for
$H_0$; this event happens with probability $1 -1/t^2 > 1 -\eps^{16}$. We condition on
this event.
}

We condition on the hash function $h$ being $S$-collision free, which happens with probability at least $1-\eps$.
Therefore, at most one variable from $H_0$ lands in each set in the partition. Since the
distribution in each partition set is 5-wise independent, this
means that the distribution on $H_0$ is fully independent. This allows
us to construct a coupling of $\bX$ and $\bY$: let $\bX_j = \bY_j$ for $j
\in H_0$, and then sample the rest according to the correct marginal
distribution. 

We say that the variables in $H_0$ are good if 
$$\Bigl|\sum_{j \in H_0} \bY_j[i] -\theta[i] \Bigr| > s \cdot
\sqrt{\sum_{j \not \in H_0} \rnorm{\bY_j[i]}_2^2} \ \text{for all} \ i
\in V$$
By Equation~\ref{junta-bound},
\begin{equation}
\label{eq:bad-1}
\Pr[\{\bX_j = \bY_j\}_{j \in H_0} \ \text{are not good}] \leq 2d\eps.
\end{equation}
We condition on these variables being good.

With this conditioning, we show that the halfspaces in $\junta$ are
nearly constant:  with high probability they do not depend on the variables outside $H_0$. To see
this, observe that conditioned on the variables in $H_0$, the
remaining variables are still 4-wise independent (in both  $\bX$ and $\bY$), so by Chebychev
\begin{equation}
\label{eq:bad-2}
\Pr\left[\Bigl|\sum_{j \not\in H_0} \bY_j[i] \Bigr| \geq s \cdot
  \sqrt{\sum_{j \not \in H_0} \rnorm{\bY_j[i]}_2^2} \right] \leq 1/s^2
  \leq \eps. 
\end{equation}
But if this does not happen, then
$$\sign(\sum_{j=1}^n \bY_j[i] - \Theta[i]) = \sign(\sum_{j \in H_0} \bY_j[i] - \Theta[i])].$$
A similar analysis holds for $\bX$. Thus for both $\bX$ and $\bY$, with
error probability at most $2d/s^2 \leq 2d\eps$, we can assume that the
halfspaces in $\junta$ are fixed to constant functions for a good choice
of variables in $H_0$. 

Recall that we are interested in fooling functions of the form
$g(h_1(\bX),\ldots,h_k(\bX))$. Conditioned on the variables in $H_0$ being
good, the halfspaces $h_j$ for $j \in \junta$ are close to constant
functions. Thus, the function $g$ is $2d\eps$ close to a function $g'$
of halfspace $\{h_j\}_{j \in \reg}$ under both distributions $\bX$ and
$\bY$. Thus it suffices to show that the bias of $g'$ under $\bX$ and
$\bY$ is close.

Conditioning on $\bX_j =\bY_j$ for $j \in H_0$ gives a halfspace
on the remaining variables in each coordinate $i \in \reg$.  
Define
$$\Theta'[i] = (\Theta[i] - \sum_{j  \not\in H_0} \bX_j[i])), \ \ \bS'[i] =
\sum_{j \not\in H_0}\bX_j[i], \ \bT'[i] = \sum_{j \not\in H_0}\bY_j[i].$$
then 
$$\sgn(\bS[i] - \Theta[i]) = \sgn(\bS'[i] - \Theta'[i]).$$ 
Thus there exists a union of orthants $A' \in
\mathbb{R}^{|\reg|}$ such that $g'(X) =1$ if $X \in A'$. Our goal is to bound
\[ \abs{\Pr[\bS' \in A'] - \Pr[\bT' \in A']}. \]

The set of random variables $\{\bX_j[i] : j \not \in H_0\}$ is $\delta$-regular. 
Hence we can apply our result for the regular case.  We've already conditioned on the
hash function $h$ being 
$H_0$-collision free.
Since this happens
with probability at least $1-\eps$, the resulting function is $\collision$-collision
preserving for $\collision=1/(1-\eps) \leq 2$, since conditioning on an event which happens with probability $p$ can increase
the probability of any other event by a factor of at most $1/p$.
So now applying the analysis from the regular case,
\begin{equation}
\label{eq:err-reg}
\abs{\Pr_{\bS'}[\bS' \in A'] - \Pr_{\bT'}[\bT' \in A']} \leq O \left(
    {\eta}^{-1/2} d^{15/8} \cdot (\frac{1}{t} + \delta)^{1/8} \right)
    \leq O(d\eps).
\end{equation}
Hence, conditioned on $\mathbf{h}$ and the variables in $H_0$ being
good, we have
\begin{equation}
\label{eq:err-reg2}
\abs{\Pr_{\bS}[\bS \in A] - \Pr_{\bT}[\bT \in A]} \leq O(d\eps) + 2d\eps.
\end{equation}

Removing the conditioning gives
\begin{align*}
\abs{\Pr_{\bS}[\bS \in A] - \Pr_{\bT}[\bT \in A]} \leq O(d\eps) +
        2d\eps + \eps + 2d\eps
= O(d\eps)
\end{align*}
\end{proof}

\section{Generalized Monotone Trick}
\label{ap:monotone}

\newcommand{\zdt}{(\zo^D)^T}
\newcommand{\udt}{\mathcal{U}}
\newcommand{\branch}{B}
\newcommand{\dbp}{D}
\newcommand{\numberbps}{d}
\newcommand{\gos}{s}

\newcommand{\acc}{\mathrm{Acc}}

We generalize the ``monotone trick'' introduced in Meka and Zuckerman \cite{MZ09}
and show that a generator that fools small-width ``monotone'' branching programs also fools
any monotone function of several arbitrary-width monotone branching programs.

First we define read-once branching programs.  Branching programs corresponding to space $S$ have
width $2^S$.  We use the following notation from \cite{MZ09}.

\begin{definition}[ROBP]\label{dfn:robp}
An $(S,\dbp,T)$-branching program $\branch$ is a layered multi-graph with a layer for each $0 \leq i \leq T$ and at most $2^{S}$ vertices (states) in each layer. The first layer has a single vertex $v_0$ and each vertex in the last layer is labeled with $0$ (rejecting) or $1$ (accepting). For $0 \leq i \leq T$, a vertex $v$ in layer $i$ has at most $2^\dbp$ outgoing edges each labeled with an element of $\{0,1\}^\dbp$ and pointing to a vertex in layer $i+1$. 
\end{definition}

Let $\branch$ be an $(S,\dbp,T)$-branching program and $v$ a vertex in layer $i$ of $\branch$. 
We now define the set of accepting suffixes.

\begin{definition}
We say $z$ is an \emph{accepting suffix} from vertex $v$ if the path in $\branch$ starting at $v$ and following
edges labeled according to $z$ leads to an accepting state.
We let $\acc_\branch(v)$ denote the set of accepting suffixes from $v$.
If $\branch$ is understood we may abbreviate this $\acc(v)$.
\end{definition}

Nisan \cite{Nis} and Impagliazzo et al.~\cite{INW} gave PRGs that fool $(S,\dbp,T)$-branching programs with error $\exp(2^{-\Omega(S+\dbp)})$ and seed length $r = O((S+\dbp + \log T)\log T)$. For $T = \poly(S,\dbp)$, the PRG of Nisan and Zuckerman \cite{NZ} fools $(S,\dbp,T)$-branching programs with seed length $r = O(S+\dbp)$.
Meka and Zuckerman showed that the above PRGs in fact fool arbitrary width branching programs of a certain form called monotone, defined next.

\begin{definition}[Monotone ROBP]\label{bpdef}
An $(S,\dbp,T)$-branching program $\branch$ is said to be monotone if for all $0 \leq i < T$, there exists an ordering $\{v_1 \prec v_2 \prec \ldots \prec v_{L_i}\}$ of the vertices in layer $i$ such that $v \prec w$ implies $\acc_\branch(v) \subseteq \acc_\branch(w)$. 
\end{definition}

Note that the natural ROBP accepting a halfspace, where states correspond to partial sums, is monotone.
However, the natural ROBP accepting the intersection of just two halfspaces may not be monotone.

The following theorem is the only known way to obtain PRGs for halfspaces using seed length which depends
logarithmically on $1/\eps$ (and polylogarithmically on $n$).

\begin{theorem}\label{thm:monotonebp} \cite{MZ09}
Let $0 < \epsilon < 1$ and $G:\{0,1\}^R \to (\{0,1\}^\dbp)^T$ be a PRG that $\delta$-fools monotone $(\log(4T/\epsilon),\dbp,T)$-branching programs. Then $G$ $(\eps+\delta)$-fools monotone $(S,\dbp,T)$-branching programs for arbitrary $S$ with error at most $\epsilon+\delta$. 
\end{theorem}

\newcommand{\up}{\mathrm{up}}
\newcommand{\down}{\mathrm{down}}
\newcommand{\bpdn}{\branch^\down}
\newcommand{\bpup}{\branch^\up}
\newcommand{\fdn}{f_\down}
\newcommand{\fup}{f_\up}

We now generalize Theorem~\ref{thm:monotonebp} to the intersection of monotone branching programs,
or even to any monotone function of monotone branching programs.
(Of course, the intersection corresponds to the monotone function AND.)

\begin{theorem}\label{thm:genmonotonebp}
Let $0 < \epsilon < 1$ and $G:\{0,1\}^R \to (\{0,1\}^\dbp)^T$ be a PRG that $\delta$-fools monotone $(\numberbps \log(4T\numberbps /\epsilon),\dbp,T)$-branching programs. Then $G$ $(\eps+\delta)$-fools any monotone function of $\numberbps $ monotone $(S,\dbp,T)$-branching programs for arbitrary $S$. 
\end{theorem}

We now generalize monotone functions to decision trees.
First note that the complement of a monotone branching program is a monotone branching program.
Now consider any decision tree, where each node of the decision tree is a monotone branching program.
Any leaf of this tree represents the intersection of monotone branching programs.
Thus, the error of the function above for such decision trees is at most $s$ times the error for each leaf.  This gives the following corollary.

\begin{corollary}\label{cor:dt}
Let $0 < \epsilon < 1$ and $G:\{0,1\}^R \to (\{0,1\}^\dbp)^T$ be a PRG that $\delta$-fools monotone $(\numberbps \log(4T\numberbps /\epsilon),\dbp,T)$-branching programs. Then $G$ ($s (\eps+\delta)$)-fools any decision tree with $s$ leaves, where each decision tree node is a monotone $(S,\dbp,T)$-branching programs for arbitrary $S$.
\end{corollary}

In the above, we can even take $s$ to be the minimum of the number of
$0$ and $1$ leaves.
We now prove Theorem~\ref{thm:genmonotonebp}, using the ideas of \cite{MZ09} based on 
``sandwiching'' monotone branching programs between small-width branching programs.

\begin{definition}
A pair of functions $(\fdn,\fup)$, each with the same domain and range as
a function $f:B \to \zo$,
is said to \emph{$\epsilon$-sandwich $f$}
if the following hold.
\begin{enumerate}
\item For all $z \in B$, $\fdn(z) \leq f(z) \leq \fup(z)$.
\item $\Pr_{z \in_u B}[\fup(z) = 1] - \Pr_{z \in_u B}[\fdn(z) = 1] \leq \epsilon$.
\end{enumerate}
\end{definition}

The following lemma shows that it suffices to fool functions which sandwich the given target function.
Bazzi \cite{Baz09} used sandwiching in showing that polylog-wise independence fools DNF formulas.
The lemma below is a small modification of a lemma in \cite{MZ09}.

\begin{lemma}
\label{lem:sandwich-easy}
If $(\fdn,\fup)$ $\epsilon$-sandwich $f$, and
a PRG $G$ $\delta$-fools $\fdn$ and $\fup$, then $G$ $(\epsilon+\delta)$-fools~$f$.
\end{lemma}

Meka and Zuckerman then showed that any monotone branching program can be sandwiched between two small-width branching programs.

\begin{lemma}
\label{lem:mzsandwich}
\cite{MZ09}
For any monotone $(S,\dbp,T)$-branching program $\branch$, there exist monotone $(\log(4T/\epsilon),\dbp,T)$-branching programs $(\bpdn,\bpup)$ that $\epsilon$-sandwich $\branch$.
\end{lemma}

Using this, we can show that any monotone function of monotone branching programs is sandwiched by a small-width branching program.

\begin{lemma}
\label{lem:genmonotonebp}
Any monotone function of $\numberbps$ $(S,\dbp,T)$-branching programs has a pair of $(\numberbps\log(4T/\epsilon),\dbp,T)$-branching programs $(\bpdn,\bpup)$ that $(\numberbps\epsilon)$-sandwich it.
\end{lemma}

\begin{proof}
For a monotone branching program~$\branch$, let $(\bpdn,\bpup)$ denote
monotone $(\log(4T/\epsilon),\dbp,T)$-branching programs that $\eps$-sandwich $\branch$, as given by Lemma~\ref{lem:mzsandwich}.
Suppose our given function is $f(z) = g(\branch_1(z),\branch_2(z),\ldots,\branch_\numberbps(z))$ for $g$ monotone.
Then $f(z)$ is sandwiched by $(\fdn,\fup)$ given by
\begin{eqnarray*}
\fdn(z) &=& f \left( \bpdn_1(z),\bpdn_2(z),\ldots,\bpdn_\numberbps(z) \right)\\
\fup(z) &=& f \left( \bpup_1(z),\bpup_2(z),\ldots,\bpup_\numberbps(z) \right).
\end{eqnarray*}
Moreover,
\[ \fup^{-1}(1) - \fdn^{-1}(1) \subseteq \bigcup_{i=1}^\numberbps \left( (\bpup_i)^{-1}(1) - (\bpdn_i)^{-1}(1) \right). \]
Since
$\Pr_{z }[\bpup_i(z) = 1] - \Pr_{z}[\bpdn_i(z) = 1] \leq \epsilon$,
it follows that
$\Pr_{z }[\fup(z) = 1] - \Pr_{z}[\fdn(z) = 1] \leq \numberbps\epsilon$.
\end{proof}

Theorem~\ref{thm:genmonotonebp} now follows from Lemmas~\ref{lem:sandwich-easy} and~\ref{lem:genmonotonebp}.
Without using any of the hard work we've done in other sections,
this theorem gives us PRGs for monotone functions of halfspaces (such as intersections) using a random seed of
length $O(\numberbps(\log n)\log (n/\eps))$.
We improve this seed length now.

\subsection{Combining the Monotone Trick and the main construction}

Fix a hash function $h$, which fixes the partition into $t$ sets.  Then any monotone function of $\vsgn(y_1W_1 + \ldots + y_nW_n - \Theta)$
may be computed by a monotone function of $d$ monotone branching programs, with $t$ layers each.  Thus, we can apply 
Theorem~\ref{thm:genmonotonebp} and Corollary~\ref{cor:dt} to deduce Theorem~\ref{thm:main-1}.

We can set $T=t$ 
and $\dbp=O(\log n)$ to store the seed for the 5-wise independent distribution.  Also note that $\log \eta^{-1} = \Theta(\log C)$.  With these parameters, using
Nisan's PRG gives a seed length of
$O((d\log(d T/\eps)+\dbp + \log T)\log T) = O((d \log(Cd/\eps) + \log n)\log(Cd/\eps))$ to fool monotone functions of $d$ halfspaces.
For functions computable by size $s$ decision trees of halfspaces, the seed length becomes
$O((d \log(Cds/\eps) + \log n)\log(Cds/\eps))$.

When $Cd/\eps \geq \log^{-c} n$ for any $c>0$, then $t=\polylog (n)$ and we can use the Nisan-Zuckerman PRG.
This gives a seed length of
$O(d\log(d T/\eps)+\dbp + \log T) = O(d \log(Cd/\eps) + \log n)$ for monotone functions of $d$ halfspaces.
For functions computable by size $s$ decision trees of halfspaces, the seed length becomes
$O(d \log(Cds/\eps) + \log n)$.

More generally, using Armoni's interpolation of Nisan and Nisan-Zuckerman will shave off an extra $\log \log n$ factor off of Nisan's PRG when
$t/\eps \leq \exp(-(\log n)^{1-\gamma})$ for some $\gamma>0$.  We omit the details.

\section{Discretizing the distribution}
\label{discretize}

The first step is to truncate each $\bx_i$ to lie in the range
$(-B,B)$.

\begin{lemma}
Set $B = (nC^2\eps^{-1})^\fr{4}$. For each $i \in [n]$,
let $\by_i = \bx_i\cdot \mb{I}(|\bx_i| < B)$. Define the product
random variable $\bY= (\by_1, \by_2, \ldots, \by_n)$ where the
$\by_i$s are independent. Then we have
\begin{itemize}
\item $ \sd(\bX,\bY) \leq \eps$.
\item $ \E[\by_i^2] \geq \fr{2}, \E[\by_i^4] \leq C$. 
\end{itemize}
\end{lemma}
\begin{proof}
Note that $\bx_i = \by_i$ when $|\bx_i| \leq B$ and $\by_i =0$
otherwise. But we have
$$\Pr[|\bx_i| \geq B] = \Pr[|\bx_i|^4 \geq B^4] \leq \frac{C}{B^4} = \frac{\eps}{nC}.$$ 
Thus it follows that 
$$\sd(\bx_i,\by_i) \leq \frac{\eps}{nC} \ \Rightarrow \ \sd(\bX,\bY) \leq \frac{\eps}{C} \leq \eps.$$

It is clear that $\E[\by_i^4] \leq \E[\bx_i^4] \leq C$. Thus we only need
to prove the claim about the two-norm. We have
$$\bx_i = \bx_i\cdot\mb{I}(|\bx_i| < B) + \bx_i \cdot\mb{I}(|\bx_i| \geq B) =
\by_i + \bx_i \cdot\mb{I}(|\bx_i| \geq B)$$
from which it follows that
$$\E[\bx_i^2] = \E[\by_i^2] + \E[\bx_i^2 \cdot\mb{I}(|\bx_i| \geq B)].$$
By the Cauchy-Schwartz inequality, we have 
$$\E[\bx_i^2 \cdot\mb{I}(|\bx_i| \geq B)] \leq
\E[\bx_i^4]^\fr{2}(\Pr[|\bx_i| \geq B)])^\fr{2} \leq
\sqrt{C}\sqrt{\frac{\eps}{nC}} = \sqrt{\frac{\eps}{n}}
< \fr{2}$$
Hence we have $\E[\by_i^2] \geq \fr{2}$.

By a similar argument, one can show that $|\E[\by_i]| \leq \frac{\eps}{nC}$. 
\end{proof}
By suitable shifting and rescaling, we can assume that the distribution satisfies
$\E[\bx_i] = 0, \E[\bx_i^2] =1$, $\E[\bx_i^4] \leq C$ and $|\bx_i| < B$
\pnote{$B$ and $C$ will bescaled by constants}.

The next step is to suitably discretize the distribution. Assume that
the random variable $\bx_i$ has a cumulative distibution function $F_i$ where
$F_i(x) = \Pr[\bx_i \leq x]$. Since $|\bx_i| < B$ we have $F(-B) =0$ and
$F(B) =1$. We will define two {\em sandwiching} discrete
  distributions $\bx_i^\ell$ and $\bx_i^u$ whose cdfs $F_i^\ell$ and
  $F^i_u$ satisfy:
$$F_i^\ell(x) \leq F_i(x) \leq F_i^\ell(x) + \gamma$$
$$F_i^u(x) -\gamma \leq F_i(x) \leq F_i^u(x)$$
where $\gamma$ is a granularity paramater (which will be chosen as
inverse polynomial in $n$).

Let $g = \fr{\gamma}$. Our goal is to define bucket boundaries $b_0,\ldots,b_g$ by
picking $b_k$ that stisfy $F_i(b_k) = k\gamma$. 

\begin{definition}
\label{def:buckets}
For $k \in \{0,\ldots,g\}$,  let $b_k$ be the
  smallest $x \in [-B,B]$ so that $F_i(x) \geq k\gamma$.
\end{definition}

\pnote{the bucket boundaries $b_k$ defined will vary for each
$\bx_i$, but we use $b_k$ in place of $b_k(i)$ for simpler notation. }

We can sample $\bx_i$ by first picking a bucket $k \in
\{0,\ldots,g-1\}$ and then sampling from this bucket according to the suitable conditional
distribution, resulting in $\bx_i \in [b_k,b_{k+1}]$.

We now define the sandwiching distributions:
\begin{definition}
\label{def:sandwich-dist}
The random variable $\bx_i^\ell$ is uniformly distributed on $\{b_0,\ldots,b_{g-1}\}$ while $\bx_i^u$ the
uniform distributed on $\{b_1,\ldots,b_g\}$. We define the family
$\mc{F}$ of $2^n$ product distributions on $\mb{R}^n$ where each co-ordinate is
distributed independently according to $\bx_i^\ell$ or $\bx_i^u$.
\end{definition}

It follows that $\sd(\bx_i^\ell,\bx_i^u) \leq \gamma$. Hence if we take
any pair of variables $\bY,\bZ$ from $\mc{F}$, by the union bound we have $\sd(\bY,\bZ) \leq \gamma n$.
The following lemma allows us to reduce the problem of fooling
halfspaces under the distribution $\bX$ to the problem of fooling a
single distribution from the family $\mc{F}$.

\begin{lemma}
\label{lem:1-halfspace}
Let $h:\mb{R}^n \rightarrow \pmo$ for $i \in [k]$ be a halfspace and 
let $\bY \in \mc{F}$. Then $$|\E[h(\bX)] - \E[h(\bY)]| \leq 4\gamma n.$$
\end{lemma}
\begin{proof}
We will pick sandwiching distributions $\bY^\ell=
(\by^\ell_1,\ldots,\by^\ell_n)$ and $\bY^u = (\by^u_1,\ldots,\by^u_n)$ from
$\mc{F}$ (depending on the halfspace $h$) and construct a coupling of
the three distributions $\bY^\ell, \bX$ and $\bY^u$ so that 
\begin{equation}
\label{eqn:couple}
h(\bY^\ell) \leq h(\bX) \leq h(\bY^u).
\end{equation}
Let $h(x) = \sgn(\sum_iw_i\bx_i - \theta)$. If $w_i \geq 0$
for all $i$, then we set
$$\by^\ell_i = \bx_i^\ell, \ \ \by_i^u = \bx_i^u.$$
Whereas if $w_i < 0$, then we set 
$$\by_i^\ell = \bx_i^u, \ \  \by_i^u = \bx_i^\ell.$$

Next we describe the coupling, co-ordinate by co-ordinate. Fix
co-ordinate $i$. Pick $k \in \{0,\ldots,g-1\}$ at random. Set
$\bx_i^\ell =  b_k$ and $\bx_i^u = b_{k+1}$. We now set the random
variables $\by_i, \by_i\ell$ and $\by_i^u$ to be eiher $\bx_i^\ell$ or
$\bx^u_i$, based on their defintion. We pick $\bx_i$ conditioned on the
$k^{th}$ bucket, so that $b_k \leq \bx_i \leq b_{k+1}$. It follows that
$$w_i\by_i^\ell \leq w_i \bx_i \leq w_i \by_i^u$$
and hence
$$\sum_i w_i\by_i^\ell \leq \sum_i w_i \bx_i \leq \sum_i w_i \by_i^u$$
which implies Equation \ref{eqn:couple}.

Since a halfspace is a statistical test, we have
\begin{equation}
\label{eq:sandwich2}
\Pr[h(\bX) \neq h(\bY^u)] \leq \Pr[h(\bY^\ell) \neq h(\bY^u)] \leq \sd(\bY^u,\bY^\ell) \leq \gamma n.
\end{equation}
If we replace $\bY^u$ with $\bY \in \mc{F}$, we have
$$\Pr[h(\bX) \neq h(\bY)] \leq \Pr[h(\bX) \neq h(\bY^u)]| + \Pr[h(\bY) \neq
  h(\bY^u)]|  \leq 2\gamma n$$
where we use Equations \ref{eq:sandwich2} and the fact that
  $\sd(\bY,\bY^u) \leq \gamma n$. The claim
  follows since $h(\bX)$ and $h(\bY)$ take values over $\pmo$.
\end{proof}

This lemma  extends to fooling functions of halfspaces.

\begin{lemma}
\label{lem:k-halfspaces}
 Let $f: \mb{R}^n \rightarrow \pmo$ be a function of $d$ halfsapces
 $h_i:\mb{R}^n \rightarrow \pmo$ given by $f =g(h_1,\ldots,h_d)$ where
 $g: \pmo^k \rightarrow \pmo$. Then for any $Y \in \mc{F}$, 
$$|\E[f(\bX)] - \E[f(\bY)]| \leq 4\gamma d n.$$
\end{lemma}
\begin{proof}
We consider the same coupling used in Lemma \ref{lem:1-halfspace}. We
have
\begin{align*}
\Pr[g(\bX) \neq g(\bY)] \leq \Pr[(h_i(\bX),\ldots,h_d(\bX)) \neq (h_1(\bY),\ldots,h_d(\bY))] \leq \sum_i
\Pr[h_i(\bX) \neq h_i(\bY)] \leq 2\gamma dn.
\end{align*}
The claim now follows since $g$ is Boolean valued.
\end{proof}

Finally, we need to show that for a suitable choice of $\gamma$, the
expectation and the second and fourth moments of $\bx_i^\ell$ and $\bx_i^u$ are nearly the
same as those of $\bx_i$. We prove the claim for the fourth moment, the
other arguments are similar.

\begin{lemma}
We have
\begin{align*}
|\E[(\bx_i)^4] - \E[(\bx_i^u)^4]|  \leq 2B^4\gamma, \ \ |\E[(\bx_i)^4] - \E[(\bx_i^\ell)^4]|  \leq 2B^4\gamma
\end{align*}
\end{lemma}
\begin{proof}
It is clear that
$$\E[(\bx_i^\ell)^4] = \gamma(\sum_{k=0}^{g-1}b_k^4), \ \ \E[(\bx_i^u)^4] = \gamma(\sum_{k=1}^{g}b_k^4).$$ 

Our goal is to compare these with the $4^{th}$ moment of $\bx_i$. The
contribution of the $k^{th}$ bucket to $\E[\bx_i^4]$ can be upper
bounded by $\gamma \max(b_k^4,b_{k+1}^4)$ and lower bounded by $\gamma
\min(b_k^4,b_{k+1}^4)$. Hence 
\begin{align*}
\gamma \sum_{k=0}^{g-1} \min(b_k^4,b_{k+1}^4) \leq \E[\bx_i^4] \leq
\gamma \sum_{k=0}^{g-1}  \max(b_k^4,b_{k+1}^4).
\end{align*}

By case analysis, the sequence $\max(b_k^4,b_{k+1}^4)$ takes on $g$ distinct values from
$\{b_0,\ldots,b_g\}$. Similarly, $\min(b_k^4,b_{k+1}^4)$ can take some
value twice but every other value at most once. Hence both the upper and lower bounds are within
$2B^4\gamma$ of both $\E[(\bx_i^\ell)^4]$ and $\E[(\bx_i^u)^4]$.
\end{proof}

A similar argument shows that the second moment changes by at most
$2B^2\gamma$ and the expectation by $2B\gamma$. We pick $\gamma < \frac{\eps}{2nB^4}
=O(\frac{\eps^2}{n^2C^2})$, which is of the form $2^{-s}$ for some
integer $s$. We have $2^s < O(\frac{n^2C^2}{\eps^2})$ hence $s =
\log(n^2C^2/\eps^2) + O(1)$. To sample from $\bx_i^\ell$ ($X^u_i$), we  pick a
random bit-string of length $s$, treat it as a number $j \in \{0,g-1\}$,
and output $b_j$ ($b_{j+1}$). 

Finally we rescale and shift, so that we again have $\E[\by_i] = 0,
\E[\by_i^2 =1]$ and $\E[\by_i^4] \leq C$.

\section{Bounded Independence fools functions of halfspaces}
\label{ap:bounded}

In this section, we prove Theorem \ref{thm:main-dwise}.

\subsection{Reduction to upper polynomials for single halfspaces}
We now flesh out the reduction described in Section \ref{sec:overview},
i.e., we show how to prove Theorem~\ref{thm:main-upper} given upper sandwiching polynomials for a single halfspace with extra properties.  

\begin{lemma} \label{lem:hybridize}  Let $\bX$ be a random vector on the product set $\Omega$, and suppose we have order-$k$ polynomials $p_1, \dots, p_d : \Omega \to \R$, as well as functions $h_1, \dots, h_d : \Omega \to \{0,1\}$.  Write $\bp = p_i(\bX)$ and $\bh_i = h_i(\bX)$.  Assume that for each $i \in [k]$:
\begin{enumerate}
\item $\bp \geq \bh_i$ with probability $1$;
\item $\E[\bp - \bh_i] \leq \eps_0$;
\item $\Pr[\bp > 1 + 1/d^2] \leq \gamma$;
\item $\rnorm{\bp}_{2d} \leq 1 + 2/d^2$.
\end{enumerate}
If we write $p = p_1 p_2 \cdots p_d$, $h = h_1 h_2 \cdots h_d$, then $p$ is a polynomial of order at most $dk$, $p(\bX) \geq h(\bX)$ with probability~$1$, and
\begin{equation} \label{eqn:hybrid-bound}
\E[p(\bX) - h(\bX)] \leq 2d \eps_0 + 3d^2\sqrt{\gamma}.
\end{equation}
\end{lemma}
\begin{proof}  The first two parts of the claim are immediate, so it suffices to verify~\eqref{eqn:hybrid-bound}.  We use the telescoping sum~\eqref{eqn:hybrid}, and thus it suffices to bound the general term as follows:
\begin{equation} \label{eqn:genl-hyb}
\E[\bh_1 \cdots \bh_{i-1} (\bp - \bh_i)\bp_{i+1} \cdots \bp_d] \leq 2\eps_0 + 3d\sqrt{\gamma}.
\end{equation}
We have
\begin{eqnarray*}
&&\E[\bh_1 \cdots \bh_{i-1} (\bp_i - \bh_i)\bp_{i+1} \cdots \bp_d] \\ 
&\leq &\E[\bp_1 \cdots \bp_{i-1} (\bp_i - \bh_i)\bp_{i+1} \cdots \bp_d] \\ 
&<& 2\E[\bp_i - \bh_i] + \E[\Ind{\bp_1 \cdots \bp_{i-1} \bp_{i+1} \cdots \bp_d \geq 2} \bp_1 \cdots \bp_{i-1} (\bp_i - \bh_i)\bp_{i+1} \cdots \bp_d] \\
&\leq& 2\eps_0 +\E\left[\left(\littlesum_{i'=1}^d \Ind{\bp_{i'} > 1+1/d^2}\right) \littleprod_{i=1}^d \bp_i\right],
\end{eqnarray*}
where in the last term we used the bounds $(1+1/d^2)^{d-1} < 2$ and $\bp_i - \bh_i \leq \bp_i$.   Thus we can establish~\eqref{eqn:genl-hyb} by showing the bound
\[
\sum_{i'=1}^d \E\left[\Ind{\bp_{i'} > 1+1/d^2} \littleprod_{i=1}^d \bp_i\right] \leq 3d \sqrt{\gamma}.
\]
This follows by bounding each summand individually:
\begin{eqnarray*}
&& \E\left[\Ind{\bp_{i'} > 1+1/d^2}\littleprod_{i=1}^d \bp_i\right]\\
&\leq& \rnorm{\Ind{\bp_{i'} > 1+1/d^2}}_2 \cdot \prod_{i=1}^d \rnorm{\bp_i}_{2d} \qquad \text{(H\"older's inequality)}\\
&\leq& \sqrt{\gamma} \cdot (1 + 2/d^2)^d \quad\leq\quad 3 \sqrt{\gamma},
\end{eqnarray*}
as needed.
\end{proof}

\subsection{Tools for upper polynomials}
We construct the upper sandwiching polynomial needed in Lemma~\ref{lem:hybridize} using two key tools:  ``DGJSV Polynomials'', the family of univariate real polynomial constructed in~\cite{DGJSV:09} for approximating the $\sgn$ function; and, our Regularity Lemma for halfspaces over general random variables~\ref{thm:crit-one}.\\

Regarding the DGJSV Polynomials, the following is a key theorem from~\cite{DGJSV:09} (slightly adjusted for our purposes):
\begin{theorem} \label{thm:dgjsv} (\cite{DGJSV:09}) Let $0 < \aaa, \bbb < 1$.  Then there exists an even integer $K = K_{\aaa,\bbb}$ with
\[
K \leq C_0 \frac{\log(2/\bbb)}{\aaa} \qquad \text{($C_0$ is a universal constant)}
\]
as well as an ordinary univariate real polynomial $P = P_{\aaa,\bbb} : \R \to \R$ of degree $K$ with the following behavior:
\begin{itemize}
\item $P(x) \geq 0$ for $x \in (-\infty, -1]$,
\item $0 \leq P(x) \leq \bbb$ for $x \in [-1, -\aaa]$;
\item $0 \leq P(x) \leq 1$ for $x \in [-\aaa, 0]$;
\item $1 \leq P(x) \leq 1+\bbb$ for $x \in [0, 1]$;
\item $P(x) \geq 1$ for $x \in [1, \infty)$;
\item $P(x) \leq (4x)^{K}$ for all $\abs{x} \geq 1$.
\end{itemize}
Note that the first five conditions imply $P(x) \geq \Ind{x \geq 0}$ for all $x \in \R$.
\end{theorem}

Regarding our Regularity Lemma for general halfspaces, we will use the following rephrasing of Theorem~\ref{thm:crit-one} with simplified parameters:

\begin{theorem}  \label{thm:reg} Let $t > 1$, $0 < \delta < 1$ and $0 < \eta$ be parameters.  Then there exists an integer $L$ satisfying
\[
L \leq \poly(\log t, 1/\eta) \cdot \frac{1}{\delta}
\]
such that the following holds. Suppose $\bx_1, \dots, \bx_n$ is a sequence of independent $\eta$-HC random variables, $\theta \in \R$, and $n \geq L$.  Then there exists a set of coordinates $H \subseteq [n]$ of cardinality $L$ such that, denoting
\[
\btheta' = \theta - \sum_{j \in H} \bx_j, \qquad \bz = \sum_{j \not \in H} \bx_j 
\]
(these random variables are independent), we have three mutually exclusive and collectively exhaustive events \emph{depending only on $\btheta'$}:
\begin{itemize}
\item Event $\mathbf{BAD}$: $\abs{\btheta'} \leq t \rnorm{\bz}_2$ and the collection $\{\bx_j : j \not \in H\}$ is not $\delta$-regular;
\item Event $\mathbf{NEAR}$: $\abs{\btheta'} \leq t \rnorm{\bz}_2$ and the collection $\{\bx_j : j \not \in H\}$ is $\delta$-regular;
\item Event $\mathbf{FAR}$: $\abs{\btheta'} > t \rnorm{\bz}_2$.
\end{itemize}
Furthermore, $\mathbf{BAD}$ has probability at most $O(1/t^4)$.
\end{theorem}

The reader will note that events $\mathbf{BAD}$, $\mathbf{NEAR}$, and
$\mathbf{FAR}$ are defined somewhat peculiarly: Neither
$\rnorm{\bz}_2$ nor the (ir)regularity of $\{\bx_j : j \not \in H\}$
is actually random.  Furthermore, by our original
Theorem~\ref{thm:crit-one}\pnote{Ryan, do we want to refer to one-dim
  crit index?}, we either have that $\{\bx_j : j \not \in H\}$
is $\delta$-regular, in which case $\mathbf{NEAR}$ and $\mathbf{FAR}$
are the only possible events, or the collection is not
$\delta$-regular, in which case $\mathbf{BAD}$ and $\mathbf{FAR}$ are
the only possible events.  Nevertheless, this tripartition of events
makes our future analysis simpler. 

\subsection{Statement of the main technical theorem, and how it completes the proof}

The main technical result we will prove is the following:
\begin{theorem} \label{thm:polyprop} Let $k \geq 1$, $0 < \delta < 1$, and $t > 4$ be parameters.  Let $\bX = 
(\bx_1, \dots, \bx_n)$ be a vector of independent $\eta$-HC random variables.  Furthermore, let $T$ be an even integer such that the $\bx_i$'s are $(T, 2, 4/t)$-hypercontractive.  Assume $T \geq C_1 d \log(dt)$, where $C_1$ is a universal constant.  Let $\theta \in \R$ and let
\[
h(x_1, \dots, x_n) = \Ind{x_1 + \cdots + x_n -\theta \geq 0}.
\]  
Then there exists a polynomial $p(x_1, \dots, x_n)$ of order $k$, with 
\[
k \leq \poly(\log t, 1/\eta) \cdot \frac{1}{\delta} + O(T/d),
\]
satisfying the $4$ properties appearing in Lemma~\ref{lem:hybridize}, with 
\[
\eps_0 = O(\sqrt{\delta})  + O(\eps_1), \qquad \eps_1 = \frac{dt \log(dt)}{T}, \qquad \gamma = 2^{-T/d}.
\]
\end{theorem}

As we now show, using Theorem~\ref{thm:polyprop} and Lemma~\ref{lem:hybridize}, we can deduce Theorem~\ref{thm:main-upper} and hence Theorem~\ref{thm:main-dwise}  simply by choosing parameters appropriately. Note that it is sufficient to prove Theorem~\ref{thm:main-upper} with $\eps \cdot \polylog(d/\eps)$ in place of $\eps$. \\

We will apply Theorem~\ref{thm:polyprop} with $\delta = \Theta(\eps^2/d^2)$ and 
\[
t = C_2 \frac{d^2}{\eps \alpha}, 
\]
where $C_2$ is a large constant of our choosing.  Regarding the hypercontractivity parameters, using
Fact \ref{fact:hyper},
we may take 
\[
\eta = \Theta(\alpha^{-1/4}), \qquad T = \Theta(t^2 \cdot \alpha \ln(2/\alpha)).
\]
The necessary assumption that 
\[
T \geq C_1 d \log(td) \quad\Leftrightarrow\quad C_2^2 \cdot \Theta\left(\frac{d^4 \ln(2/\alpha)}{\eps^2 \alpha}\right) \geq C_1 d \log\left(C_2 \frac{d^3}{\eps\alpha}\right)
\]
is valid provided that $C_2$ is a sufficiently large constant.\\

We obtain from the theorem an upper $\eps_2$-sandwiching polynomial for $h$ with order 
\[
k = \wt{O}(d^2/\eps^2) \cdot \poly(1/\alpha) + O(d^3/\eps^2) \cdot \poly(1/\alpha) \leq O(d^3/\eps^2) \cdot \poly(1/\alpha),
\]
where
\[
\eps_0 = O(\eps/d) + \wt{O}(\eps/d) = \wt{O}(\eps/d)
\]
and $\gamma$ is exponentially small in $d/(\eps \alpha)$.  By using such polynomials in Lemma~\ref{lem:hybridize}, we get upper sandwiching polynomials for intersections of $d$ halfspaces with the claimed degree $kd = \wt{O}(d^4/\eps^2) \cdot \poly(1/\alpha)$ and the claimed error $d\eps_0 = \eps \cdot \polylog(d/\eps)$.

\subsection{Proof of Theorem~\ref{thm:polyprop}}
In this section, we prove Theorem~\ref{thm:polyprop}.  Let $H$ be the set of cardinality $L = \poly(\log t, 1/\eta) \cdot (1/\delta)$  coming from Theorem~\ref{thm:reg}, and assume without loss of generality that $H = \{1, \dots, L\}$.  We use the notation $\theta' = \theta - (x_1 + \cdots + x_L)$, $z = x_{L+1} + \cdots + x_n$, $\mathrm{BAD} = \mathrm{BAD}(x_1, \dots, x_L)$ etc., with boldface indicating randomness as usual.  Given the outcomes for $\bx_1, \dots, \bx_L$, we will handle the three events $\mathbf{BAD}$, $\mathbf{NEAR}$, and $\mathbf{FAR}$ with separate ordinary real polynomials. More precisely, our final (generalized) polynomial will be
\[
p(x_1, \dots, x_n) = \Ind{\mathbf{BAD}}\cdot 1 +\Ind{\mathbf{NEAR}} \cdot p^{\mathrm{near}}_{\theta'}(z) + \Ind{\mathbf{FAR}}\cdot p^{\mathrm{far}}_{\theta'}(z),
\]
where
\[
p^{\mathrm{near}}_{\theta'}(z) = P\left(\frac{z-\theta'}{2t\rnorm{\bz}_2}\right),
\]
and
\[
p^{\mathrm{far}}_{\theta'}(z) =\Ind{\theta' > 0 } \cdot1 + \Ind{\theta' \leq 0} \cdot \left(\frac{z}{\theta'}\right)^{q},
\]
where $q$ is a positive integer and $P$ is an ordinary real univariate polynomial to be specified later.  For typographic simplicity, we will write simply $p_{\theta'}$ in place of $p^{\mathrm{near}}_{\theta'}$ and $p^{\mathrm{far}}_{\theta'}$, with context dictating which we are referring to.\\

Let us walk through the properties of $p$ we need to prove.  Regarding its order, we will prove that both
\[
q \leq O(T/d), \qquad \deg P \leq O(T/d);
\]
i.e.\ when $\theta'$ is fixed, $p_{\theta'}(x_{L+1}, \dots, x_{n})$ has degree at most $O(T/d)$ as an ordinary multivariate real polynomial.  Since $\theta'$, $\mathrm{BAD}$, $\mathrm{NEAR}$, and $\mathrm{FAR}$ are determined by $x_1, \dots, x_L$ alone, it follows that our final polynomial $p$ is a generalized polynomial of order at most $L + O(T/d)$, as needed for the theorem.\\

Next, we discuss Condition~1, that $p(\bX) \geq h(\bX)$ always.  For the $\mathrm{BAD}$ outcomes for $\bx_1, \dots, \bx_L$ we have $p(\bX) = 1 \geq h(\bX)$.  For the remaining outcomes, we will have $p(\bX) \geq h(\bX)$ as required provided that in all cases
\begin{equation} \label{eqn:show1}
p_{\theta'}(\bz) \geq h_{\theta'}(\bz) \qquad \text{for all $\theta'$ and $z$}
\end{equation}
where
\[
h_{\theta'}(\bz) = \Ind{\bz - \theta' \geq 0}.
\]

Next, we discuss Condition~2, the bound $\E[p(\bX) - h(\bX)] \leq \eps_1$.  It suffices to prove an upper bound of $O(\eps_1)$.  Recall that
\[
\eps_1 = \frac{dt \log(dt)}{T}.
\]
Note also that we will always $T \leq t^2$, since no random variable has stronger hypercontractivity than do Gaussians, for which $T \leq 1 + t^2/16$.  It follows that we will always have $\eps_1 \geq 1/t$.  Thus the probability of $\mathbf{BAD}$, which is at most $O(1/t^4)$, is much smaller than $O(\eps_1)$ and can therefore be neglected.  Hence it suffices to show that
\begin{equation} \label{eqn:show2}
\E[p_{\theta'}(\bz) - h_{\theta'}(\bz)] \leq O(\eps_1)
\end{equation}
holds in both of the following cases:
\paragraph{Case Near:} $\abs{\theta'} \leq t \rnorm{\bz}_2$ and the collection $\{\bx_{L+1}, \dots, \bx_n\}$ is $\delta$-regular.
\paragraph{Case Far:} $\abs{\theta'} > t \rnorm{\bz}_2$.\\

Next we discuss Condition~3, the bound $\Pr[p(\bX) > 1 + 1/d^2] \leq 2^{-T/d}$.  Again, since $p(\bX) = 1$ for the bad outcomes $x_1, \dots, x_L$, it suffices to show that 
\begin{equation} \label{eqn:show3}
\Pr[p_{\theta'}(\bz) > 1 + 1/d^2] \leq 2^{-T/d}
\end{equation}
holds in both Case~a and Case~b.  \\

Finally, we discuss the bound $\rnorm{p(\bX)}_{2d} \leq 1 + 2/d^2$. We have
\[
\E[p(\bX)^{2k}] \leq (1 + 1/d^2)^{2d} + \E[p(\bX)^{2d} \cdot \Ind{p(\bX) > 1 + 1/d^2}] \leq 1+3/d + \E[p(\bX)^{2d} \cdot \Ind{p(\bX) > 1 + 1/d^2}].
\]
If we can show that
\[
\E[p(\bX)^{2d} \cdot \Ind{p(\bX) > 1 + 1/d^2}] \leq 1/d, 
\]
then we will have shown
\[
\E[p(\bX)^{2d}] \leq 1 + 4/d \leq (1 + 2/d^2)^{2d},
\]
as required.  Thus it remains to establish the previous upper bound.  Again, since $p(\bX) = 1$ for the BAD outcomes $x_1, \dots, x_L$, it suffices to show that 
\begin{equation} \label{eqn:show4}
\E[p_{\theta'}(\bz)^{2d} \cdot \Ind{p_{\theta'}(\bz) > 1 + 1/d^2}] \leq 1/d
\end{equation} 
holds in both Case~Near and Case~Far.

Summarizing, our goal is to construct univariate polynomials
$p_{\theta'}(z)$ of degree at most $O(T/d)$ for each of Case~Near and
Case~Far so that~\eqref{eqn:show1}, \eqref{eqn:show2},
\eqref{eqn:show3}, and~\eqref{eqn:show4} all hold.  We will first
handle Case~Near, the more difficult case.

\subsubsection{Case Near}\label{sec:case2}
In this case we have $\abs{\theta'} \leq t \rnorm{\bz}_2$, where $\bz = \bx_{L+1} + \cdots + \bx_n$ is the sum of a $\delta$-regular collection of independent random variables.  Our task is to construct a real polynomial $p_{\theta'}(z)$ of degree at most $O(T/d)$ such that bounds~\eqref{eqn:show1}, \eqref{eqn:show2}, \eqref{eqn:show3}, and~\eqref{eqn:show4} all hold with respect to the function $h_{\theta'}(z) = \Ind{z - \theta' \geq 0}$.\\

Given the parameters $d$ and $t$, choose
\[
\aaa = 16 C_0 \frac{d \log(td)}{T}, \qquad \bbb = \min(1/d^2, 1/t^4);
\]
we have $\aaa < 1$ assuming that the $C_1$ in our assumption on $T$ is large enough.  Let $K = K_{\aaa,\bbb}$ and $P=P_{\aaa,\bbb}$ be the resulting even integer and univariate polynomial from Theorem~\ref{thm:dgjsv}.  Our choice of $\aaa$ was arranged so that
\begin{equation} \label{eqn:alpha}
K \leq \frac{T}{4d}.
\end{equation}

We will define
\[
p_{\theta'}(\bz) = p_{\mathrm{near}}(\theta',\bz)=P(\bw), \qquad \text{where } \bw = \frac{\bz-\theta'}{2t\rnorm{\bz}_2}.
\]

Thus $p_{\theta'}(z)$ has degree $K = O(T/d)$  as necessary, and it also satisfies~\eqref{eqn:show1}, using the property that $P \geq 0$ on $(-\infty, 0]$ and $P \geq 1$ on $[0, \infty)$.

Next we check~\eqref{eqn:show4}.
i.e.,
\[
\E[p_{\theta'}(\bz)^{2d} \cdot \Ind{p_{\theta'}(\bz) > 1 + 1/d^2}] \leq 1/d.
\]
Since $\bbb \leq 1/d^2$, we have that $p_{\theta'}(\bz) > 1+1/d^2$ only if $\abs{\bw} \geq 1$.

Also notice that $p_{\theta'}(z)\leq (4w)^K$, it suffice to bound 
$\E[\Ind{|w|\geq1} \cdot (4w)^{2dK}]$ and we will prove a stronger result:
\begin{equation}\label{eqn:strong}
        \E[(4w)^{2dK} \cdot \Ind{|w|\geq1}]\leq 2^{-T}. \end{equation}

To see this, since we are in Case~Near we have $\abs{\theta'} < t\rnorm{\bz}_2$.  Thus if $|w|\geq 1$, we must have $\abs{\bz} > t\rnorm{\bz}_2$.  This also implies $\abs{\bz - \theta'} < 2\abs{\bz}$; hence we have 
\[
\abs{4\bw} = 2\frac{\abs{\bz - \theta'}}{t \rnorm{\bz}_2} < \frac{4}{t} \cdot \abs{\frac{\bz}{\rnorm{\bz}_2}}.
\]
Thus we have\begin{eqnarray}
 &&\E\left[\Ind{|w|\geq 1} \cdot (4w)^{2dK}\right] \nonumber\\ &\leq& \E\left[\Ind{\abs{\bz} > t\rnorm{\bz}_2} \cdot \left(\frac{4}{t}\right)^{2dK} \left(\frac{\bz}{\rnorm{\bz}_2}\right)^{2dK}\right] \nonumber\\
&=& \left(\frac{4}{t}\right)^{2dK} \cdot \E\left[\Ind{\abs{\frac{\bz}{\rnorm{\bz}_2}} > t} \cdot  \left(\frac{\bz}{\rnorm{\bz}_2}\right)^{2dK}\right]. \label{eqn:largedev}
\end{eqnarray}

It is easy to check that
\[
\Ind{\abs{\frac{\bz}{\rnorm{\bz}_2}} > t} \cdot  \left(\frac{\bz}{\rnorm{\bz}_2}\right)^{2dK} \leq \left(\frac{\bz}{t\rnorm{\bz}_2}\right)^T \cdot t^{2dK},
\]
using the fact that $2dK \leq T$. Thus we may upper-bound~\eqref{eqn:largedev} by
\[
4^{2dK} t^{-T} \frac{\rnorm{\bz}_T^T}{\rnorm{\bz}_2^T} \leq 4^{2dK} t^{-T} (t/4)^T = 4^{2dK-T},
\]
where we used the $(T,2,4/t)$-hypercontractivity of $\bz$.  Since we have
\begin{equation} \label{eqn:kkt}
2dK \leq T/2,
\end{equation}
by virtue of~\eqref{eqn:alpha}, we conclude
\begin{equation} \label{eqn:caseiii}
\E[p_{\theta'}(\bz)^{2d} \cdot \Ind{p_{\theta'}(\bz) > 1 + 1/d^2}]\leq 4^{-T/2} = 2^{-T} \leq 1/d.
\end{equation}

Let us move on to showing~\eqref{eqn:show3} in this Case~Near; i.e., upper-bounding $\Pr[p_{\theta'}(\bz) > 1 + 1/d^2]$.  Since $\bbb \leq 1/d^2$, again we have that $p_{\theta'}(\bz) > 1/d^2$ only if $\abs{\bw} \geq 1$.  But by~\eqref{eqn:strong} \rnote{would have been better to separate this out/make it clearer}
\[
\E[\Ind{\abs{\bw} \geq 1} \cdot (4\bw)^{dK}] \leq 2^{-T},
\]
and the left-hand side is clearly an upper bound on $\Pr[\abs{\bw} \geq 1]$.  Thus we have established~\eqref{eqn:show3} in Case~Near.\\ 

Last, we will work to upper bound $\E[p_{\theta'}(\bz) - h_{\theta'}(\bz)]$ so as to show~\eqref{eqn:show2} in Case~Near.  We analyze three subcases, depending on the magnitude of $\bw$.

\paragraph{Case i: $-\aaa \leq \bw \leq 0$.}  In this case, we upper-bound $p_{\theta'}(\bz) - h_{\theta'}(\bz)$ simply by~$1$, and argue that Case~i occurs with low probability. Specifically, 
\[
\Pr[-\aaa \leq \bw \leq 0] \leq \Pr[\abs{\bw} \leq \aaa] = \Pr[\abs{\bz - \theta'} \leq 2t\aaa \cdot \rnorm{\bz}_2].
\]
We can upper-bound this probability using the Berry-Esseen
Theorem~\cite[Corollary 4.5]{MZ09}.  Since we have $\delta$-regularity
of $\bx_{L+1}, \dots, \bx_{n}$ in Case~Near, we get 
\[
\Pr[\abs{\bz - \theta'} \leq 2t\aaa \cdot \rnorm{\bz}_2] \leq O(\sqrt{\delta} + t\aaa)
\]

By definition of $\aaa$ we have $O(t\aaa) = O(\eps_1)$.  Thus we conclude for Case~i, 
\begin{equation} \label{eqn:casei}
\E[\Ind{\text{Case~i}} \cdot (p_{\theta'}(\bz) - h_{\theta'}(\bz))] \leq O(\sqrt{\delta} + \eps_1).
\end{equation} 

\paragraph{Case ii: $\abs{\bw} \leq 1$ but not Case~i.}  In this case, we have $p_{\theta'}(\bz) - h_{\theta'}(\bz) \leq \bbb \leq 1/t^4$, by construction.  Thus
\begin{equation} \label{eqn:caseii}
\E[\Ind{\text{Case~ii}} \cdot (p_{\theta'}(\bz) - h_{\theta'}(\bz))] \leq 1/t^4 \leq O(\eps_1).
\end{equation} 

\paragraph{Case iii: $\abs{\bw} > 1$.}  I.e., $\abs{\bz - \theta'} > 2t\rnorm{\bz}_2$. Notice that $p_{\theta'}(\bz) - h_{\theta'}(\bz)\leq p_{\theta'}(\bz)$  and therefore 
 $$\E[\Ind{\text{Case~iii}} \cdot (p_{\theta'}(\bz) - h_{\theta'}(\bz))]\leq \E[p_{\theta'}(z) \cdot \Ind{\abs{\bw} \geq 1}] \leq \E\left[\Ind{|w|\geq 1} (4w)^{dK}\right]\leq 2^{-T}\leq O(\eps_1)$$ (the second last inequality is due to  \eqref{eqn:strong}).

\subsubsection{Case Far}
If $\theta < 0$ then $h_{\theta'}$ is almost always~$1$.  As stated, in this case we simply have $p_{\theta'}(z) \equiv 1$.  Bounds~\eqref{eqn:show1}, \eqref{eqn:show3}, and~\eqref{eqn:show4} become trivial; for~\eqref{eqn:show2} it suffices to show
\begin{equation} \label{eqn:easier}
\Pr[\bz \leq \theta'] \leq \eps_1.
\end{equation}
We will show a stronger statement in the course of handling the case that $\theta' > 0$.\\

So it remains to handle the $\theta' > 0$ case.  As stated, in this case we define
\[
p_{\theta'}(z) = p_{\mathrm{far}}(\theta',z)= \left(\frac{z}{\theta'}\right)^q,
\]
where
\[
q = \left\lfloor\frac{T}{2d}\right\rfloor_{\text{even}},
\]
meaning $T/2d$ rounded down to the nearest even integer.  Note that $p_{\theta'}(z)$ has the claimed degree bound $O(T/d)$ (treating $\theta'$ as a constant).  Also note that $p_{\theta'}(\bz) \geq 1$ if and only if $\abs{\bz} \geq \theta'$.  This establishes~\eqref{eqn:show1}. \\

Let's move to~\eqref{eqn:show3}; we need
\[
\Pr[p_{\theta'}(\bz) \geq 1 + 1/d^2] \leq 2^{-T/d}.
\]
Certainly
\[
p_{\theta'}(\bz) \geq 1 + 1/d^2  \quad\Rightarrow\quad p_{\theta'}(\bz) \geq 1 \quad\Rightarrow\quad \abs{\bz} \geq \abs{\theta'}.
\]
It thus suffices to show
\[
\Pr[\abs{\bz} \geq \abs{\theta'}] \leq 2^{-T/d}, 
\]
which, once shown, also establishes~\eqref{eqn:easier}, since $2^{-T/d} \ll \eps_1$. We will in fact show the stronger statement
\begin{equation} \label{eqn:stronger}
\E\left[\left(\frac{\bz}{\theta'}\right)^q\right] \leq 2^{-T/d}.
\end{equation}

And this stronger statement establishes~\eqref{eqn:show2}, again because $2^{-T/d} \leq \eps_1$.\\

To prove~\eqref{eqn:stronger} we appeal to the condition of Case~Far, $\abs{\theta'} > t \rnorm{\bz}_2$.  Thus
\begin{eqnarray*}
\E\left[\left(\frac{\bz}{\theta'}\right)^q\right] & \leq & \E\left[\left(\frac{\bz}{t \rnorm{\bz}_2}\right)^q\right] \\
& \leq & \E\left[\left(\frac{\bz}{t \rnorm{\bz}_2}\right)^T\right]^{q/T} \qquad \text{(Jensen, since $T/q \geq 1$)} \\
& = & t^{-q} \left(\frac{\rnorm{\bz}_T^T}{\rnorm{\bz}_2^T}\right)^{q/T} \\
& \leq & t^{-q} \left(\frac{t}{4}\right)^{q} \qquad \text{(by $(T,2,4/t)$-hypercontractivity of $\bz$)} \\
& = & 4^{-q} \quad=\quad 2^{-T/d},
\end{eqnarray*}
using the definition of $q$.\\

Finally, to prove~\eqref{eqn:show4} it certainly suffices to show
\[
1/d \geq \E[p_{\theta'}(\bz)^{2d}] = \E\left[\left(\frac{\bz}{\theta'}\right)^{2d}\right].
\]
By repeating the previous inequality with $2d$ in place of $q$ (we
still have $T/2d \geq 1$), we can upper-bound the expectation by
$4^{-2d}$, which is indeed at most $1/d$.  This concludes the
verification of Case~Far, and thus all of Theorem~\ref{thm:polyprop}.

\section{  Fooling the uniform distribution on the sphere}
\label{ap:sphere}

In this section, we will show that our PRG can also be used to fool
any function of $d$ halfspaces over the uniform distribution on the
$n$ dimensional unit sphere; building such a PRG also has an
application in derandomizing the hardness of learning reduction
in~\cite{KS09}. 

The main idea  is to  show that the $n$ dimensional Gaussian
distribution can be use to fool the uniform distribution on the
sphere. Therefore, it suffice to fool the $n$ dimensional Gaussian
which is studied  in the previous sections (either using the modified MZ generator or $k$-wise independence). 

 Specifically, we first show the following connection between  the $n$ dimensional Gaussian distribution $\calN(0,1/\sqrt{n})^n$ and the uniform distribution on the $n$ dimensional unit sphere $S_{n-1}$.
\begin{lemma}~\label{lem:sphere}For any $\theta_1,\theta_2,..\theta_d\in \R$ and $W_1,W_2,..W_d\in \R^n$ and $h_i(X)=\sgn(W_i\cdot X-\theta_i)$ and $f:\{0,1\}^d\to \{0,1\}$, there is some universal constant C such that 
\begin{equation}\label{eqn:sphere}
\big|\E_{\bX\in _uS^{n-1}}[f(h_1(\bX),..,h_d(\bX))] - \E_{\bX\in _u\calN(0,1/\sqrt{n})^n}[f(h_1(\bX),h_2(\bX)..h_d(\bX)]\big|\leq \frac{Cd\log n}{n^{{1}/{4}}}
\end{equation}
\end{lemma}

\begin{proof}
Notice that if we  choose $x\in _u\calN(0,1/\sqrt{n})^n$, then  $\frac{x}{\|x\|_2}$ follows the uniform distribution on the sphere.
Therefore, we only need to bound:
\begin{multline} \label{eq:tob}
\big|\E_{\bX\in _u\calN(0,1/\sqrt{n})^n}(f(h_1(\frac{\bX}{\|\bX\|_2}),..,h_d(\frac{\bX}{\|\bX\|_2})) - \E_{x\in _u\calN(0,1/\sqrt{n})^n}f(h_1(\bX),h_2(\bX)..h_d(\bX))|
\\\leq \Pr_{x\in _u\calN(0,1/\sqrt{n})^n}\big(f(h_1(\frac{\bX}{\|\bX\|_2}),..,h_d(\frac{\bX}{\|\bX\|_2}))\neq f(h_1(\bX),h_2(\bX)..h_d(\bX))\big)\\
\leq \sum_{i=1}^d \Pr_{x\in _u\calN(0,1/\sqrt{n})^n}(h_i(\frac{\bX}{\|\bX\|_2})\ne h_i(\bX))
\end{multline}

By Lemma 6.2 in~\cite{MZ09}, we know that:$$\Pr_{\bX\in _u\calN(0,1/\sqrt{n})^n}(h_i(\frac{\bX}{\|\bX\|_2})\ne h_i(x)) \leq \frac{C\log n}{n^{1/4}}.$$   Combining above inequality with~\eqref{eq:tob}, we prove~\eqref{eqn:sphere}. 
\end{proof}

Therefore to fool any function of $d$ halfspaces over the uniform distribution on the $n$ dimensional sphere with accuracy $\Omega(\frac{C\log n}{n^{1/4}})$, it suffice to build a PRG for $n$ dimensional Gaussian distribution with the same accuracy.

\subsection{Derandomized hardness of learning intersections of halfspaces}

One of the application of above PRG is that we can use it to derandomize the  hardness of learning result in~\cite{KS09}.
In~\cite{KS09}, Khot and Saket showed that assuming NP$\ne$RP, for any $\eps >0$ and positive integer $d$, given a set of examples such that there is a intersection of two halfspaces that is consistent with all the examples, it is NP-hard to find a function of any $d$ halfspaces that is consistent with a $1/2+O(\eps)$ fraction of the examples.
Our PRGs can be used to derandomize the hardness reduction and  obtain the same hardness result  assuming NP$\ne$ P.

 To see why our PRG works, we need to look into the details
 of~\cite{KS09}. Let us explain in high level why our PRG  helps,
 without entering into the details of the reduction. The hardness of
 learning result in~\cite{KS09} is based on a reduction  from a Label
 Cover instance $\calL$ to  a distribution $\calD_{0}$ on negative
 examples and a distribution  $\calD_1$ on positive examples. Such a
 reduction would preserve the following two properties: 
\begin{itemize}
\item (Completeness) if the optimum value of $\calL$  is 1, then there is a intersection of  two halfspaces $f(x)$ that agrees with all the examples; i.e., $\E_{\calD_1}[f(\bX)]=\E_{\calD_0}[f(\bX)]+1$.
\item (Soundness) if the optimum value of $\calL$ is small,
  then for any $h(x)$ which is a function of $d$ halfspaces, we have
  that $\big|\E_{D_0}[h(\bX)]-\E_{D_1}[h(\bX)]\big|=O(\eps)$ which implies that
  $h(x)$ agrees with at most $1/2+O(\eps)$ fraction of the examples.  
\end{itemize}

The  $D_i$  for $(i=0,1)$ constructed in~\cite{KS09} is  a  mixture of uniform distribution on the sphere located at different center and the number of the different spheres is $\poly(n)$, where $n$ is the size of the Label Cover instance. Then by the PRG in this paper, we can derandomize each sphere with some distribution that only has support of size $\poly(n)$ to $\eps$-fool functions of $d$ halfspaces; and overall we can get distribution $\calP_0$ and $\calP_1$ with $\poly(n)$ support and it has the  property that for any function $h(x)$ of $l$ halfspaces, $|\E_{\calD_i}[f(\bX)]-\E_
{\calP_i}[f(\bX)]|\leq O(\eps)$ for $i = 0,1.$  If we replace $\calD_i$ with $\calP_i $ in the hardness reduction, we still get the soundness guarantee that $|\E_{P_1}[f(\bX)]-\E_{P_0}[f(\bX)]|=O(\eps).$ 

 We also need to verify that the  completeness property will hold if
 we replace $\calD_i$ with $\calP_i.$ If we look into the reduction
 of~\cite{KS09}, as long as the distribution $\calP_i$  has all its
 support points on the sphere, the reduction will preserve the
 completeness property. Therefore, to make the reduction work, we need
 to build a PRG for functions of $d$-halfspaces over the  uniform
 distribution on the sphere with the additional property that all the
 points generated by the PRG are all on the unit sphere as well.  

This is also achievable and  we summarize the high level idea here.
As is shown in Lemma~\ref{lem:sphere}, it suffice to fool functions of
$d$ halfspaces over  $n$ dimensional Gaussian instead of the uniform
distribution on the sphere. In addition, by the proof of
Theorem~\ref{thm:very-long}, if we only want to fool any functions of
$d$ $\eps$-regular halfspaces, it suffice just to fool  uniform
distribution on $\{-1/\sqrt{n},1/\sqrt{n}\}^n$ instead. For the
uniform distribution over $\{-1/\sqrt{n},1/\sqrt{n}\}^n$. we know that
it can be fooled by PRG with all the support points in
$\{-1/\sqrt{n},1/\sqrt{n}\}^n$ which is a subset of the unit sphere. 
 To handle the case that $d$ halfspaces are not all $\eps$-regular, we
 can follow~the idea of \cite{MZ09} Lemma 6.3 by showing that there
 exists a set of $\poly(n)$ unitary rotations  and with high
 probability that all of the $d$ halfspaces  become regular under a
 rotation randomly chosen from the set.

\bibliography{everything}

\begin{thebibliography}{FGRW09}

\bibitem[Baz09]{Baz09}
L.~Bazzi.
\newblock Polylogarithmic independence can fool {DNF} formulas.
\newblock {\em SIAM Journal on Computing}, 38:2220--2272, 2009.

\bibitem[BELY09]{BLY:09}
Ido Ben-Eliezer, Shachar Lovett, and Ariel Yadin.
\newblock Polynomial threshold functions: Structure, approximation and
  pseudorandomness.
\newblock In {\em Submitted}, 2009.

\bibitem[Ben04]{Ben04}
Vidmantas Bentkus.
\newblock A {L}yapunov type bound in $\mathbb{R}^d$.
\newblock {\em Theory of Probability and its Applications}, 49(2):311--322,
  2004.

\bibitem[Der65]{Dertouzos:65}
M.~Dertouzos.
\newblock {\em Threshold logic: a synthesis approach}.
\newblock MIT Press, Cambridge, MA, 1965.

\bibitem[DGJ{\etalchar{+}}09]{DGJSV:09}
Ilias Diakonikolas, Parikshit Gopalan, Ragesh Jaiswal, Rocco~A. Servedio, and
  Emanuele Viola.
\newblock Bounded independence fools halfspaces.
\newblock In {\em Proceedings of the 50th IEEE Symposium on Foundations of
  Computer Science}, 2009.

\bibitem[DKN09]{DKN:09}
I.~Diakonikolas, D.~Kane, and J.~Nelson.
\newblock Bounded independence fools degree-2 threshold functions.
\newblock In {\em Submitted}, 2009.

\bibitem[DS79]{DubeyShapley:79}
P.~Dubey and L.S. Shapley.
\newblock Mathematical properties of the banzhaf power index.
\newblock {\em Mathematics of Operations Research}, 4:99--131, 1979.

\bibitem[FGRW09]{FGRW:09}
V.~Feldman, V.~Guruswami, P.~Raghavendra, and Y.~Wu.
\newblock Agnostic learning of monomials by halfspaces is hard.
\newblock In {\em FOCS}, 2009.

\bibitem[FKL{\etalchar{+}}01]{FKL+:01}
J.~Forster, M.~Krause, S.V. Lokam, R.~Mubarakzjanov, N.~Schmitt, and H.-U.
  Simon.
\newblock Relations between communication complexity, linear arrangements, and
  computational complexity.
\newblock In {\em FSTTCS}, pages 171--182, 2001.

\bibitem[GR09]{GR:09}
P.~Gopalan and J.~Radhakrishnan.
\newblock Finding duplicates in a data stream.
\newblock In {\em Proc.\ 20th Annual Symposium on Discrete Algorithms
  (SODA'09)}, pages 402--411, 2009.

\bibitem[HKM09]{HKM:09}
Prahladh Harsha, Adam Klivans, and Raghu Meka.
\newblock An invariance principle for polytopes.
\newblock In {\em Submitted}, 2009.

\bibitem[HMP{\etalchar{+}}93]{HMP+:87}
A.~Hajnal, W.~Maass, P.~Pudlak, M.~Szegedy, and G.~Turan.
\newblock Threshold circuits of bounded depth.
\newblock {\em Journal of Computer and System Sciences}, 46:129--154, 1993.

\bibitem[Hu65]{Hu:65}
S.T. Hu.
\newblock {\em Threshold Logic}.
\newblock University of California Press, 1965.

\bibitem[INW94]{INW}
Russell Impagliazzo, Noam Nisan, and Avi Wigderson.
\newblock Pseudorandomness for network algorithms.
\newblock In {\em STOC}, pages 356--364, 1994.

\bibitem[Isb69]{Isbell:69}
J.R. Isbell.
\newblock {A Counterexample in Weighted Majority Games}.
\newblock {\em Proceedings of the AMS}, 20(2):590--592, 1969.

\bibitem[Kra91]{Krause:91}
M.~Krause.
\newblock Geometric arguments yield better bounds for threshold circuits and
  distributed computing.
\newblock In {\em Proc. 6th Structure in Complexity Theory Conference}, pages
  314--322, 1991.

\bibitem[KS88]{KS88}
Wies{\l}aw Krakowiak and Jerzy Szulga.
\newblock Hypercontraction principle and random multilinear forms.
\newblock {\em Probability Theory and Related Fields}, 77(3):325--342, 1988.

\bibitem[KS08]{KS09}
S.~Khot and R.~Saket.
\newblock On hardness of learning intersection of two halfspaces.
\newblock In {\em STOC}, 2008.

\bibitem[KW91]{KrauseWaack:91}
M.~Krause and S.~Waack.
\newblock Variation ranks of communication matrices and lower bounds for depth
  two circuits having symmetric gates with unbounded fanin.
\newblock In {\em Proc.\ 32nd IEEE Symposium on Foundations of Computer Science
  (FOCS)}, pages 777--782, 1991.

\bibitem[LC67]{LewisCoates:67}
P.M. Lewis and C.L. Coates.
\newblock {\em Threshold Logic}.
\newblock New York, Wiley, 1967.

\bibitem[MOO05]{MOO05}
Elchanan Mossel, Ryan O'Donnell, and Krzysztof Oleszkiewicz.
\newblock Noise stability of functions with low influences: invariance and
  optimality.
\newblock In {\em Proceedings of the 46th IEEE Symposium on Foundations of
  Computer Science}, pages 21--30, 2005.
\newblock To appear, \emph{Annals of Mathematics} 2010.

\bibitem[Mos08]{Mos08}
Elchanan Mossel.
\newblock Gaussian bounds for noise correlation of functions and tight analysis
  of long codes.
\newblock In {\em Proceedings of the 49th IEEE Symposium on Foundations of
  Computer Science}, pages 156--165, 2008.

\bibitem[Mur71]{Muroga:71}
S.~Muroga.
\newblock {\em Threshold logic and its applications}.
\newblock Wiley-Interscience, New York, 1971.

\bibitem[MZ09]{MZ09}
Raghu Meka and David Zuckerman.
\newblock Pseudorandom generators for polynomial threshold functions, 2009.
\newblock arXiv:0910.4122 [cs.CC].

\bibitem[Nis92]{Nis}
Noam Nisan.
\newblock Pseudorandom generators for space-bounded computation.
\newblock {\em Combinatorica}, 12(4):449--461, 1992.

\bibitem[NZ96]{NZ}
Noam Nisan and David Zuckerman.
\newblock Randomness is linear in space.
\newblock {\em Journal of Computer and System Sciences}, 52(1):43--52, 1996.

\bibitem[OS08]{OS:08}
R.~O'Donnell and R.~Servedio.
\newblock {The Chow Parameters Problem}.
\newblock In {\em STOC}, pages 517--526, 2008.

\bibitem[Pen46]{Penrose:46}
L.S. Penrose.
\newblock The elementary statistics of majority voting.
\newblock {\em Journal of the Royal Statistical Society}, 109(1):53--57, 1946.

\bibitem[RS08]{RS:08}
Y.~Rabani and A.~Shpilka.
\newblock Explicit construction of a small epsilon-net for linear threshold
  functions.
\newblock In {\em STOC}, 2008.

\bibitem[Ser07]{Ser:07}
R.~Servedio.
\newblock {Every linear threshold function has a low-weight approximator}.
\newblock {\em Computational Complexity}, 16(2):180--209, 2007.

\bibitem[She69]{Sheng:69}
Q.~Sheng.
\newblock {\em Threshold Logic}.
\newblock London, New York, Academic Press, 1969.

\bibitem[Szu90]{Szu90}
Jerzy Szulga.
\newblock A note on hypercontractivity of stable random variables.
\newblock {\em The Annals of Probability}, 18(4):1746--1758, 1990.

\bibitem[TZ92]{TaylorZwicker:92}
A.~Taylor and W.~Zwicker.
\newblock {A Characterization of Weighted Voting}.
\newblock {\em Proceedings of the AMS}, 115(4):1089--1094, 1992.

\bibitem[Wol06a]{Wol06}
Pawe{\l} Wolff.
\newblock Hypercontractivity of random variables and geometry of linear normed
  spaces, 2006.
\newblock Unpublished.

\bibitem[Wol06b]{Wol07}
Pawe{\l} Wolff.
\newblock Hypercontractivity of simple random variables.
\newblock {\em Studia Mathematica}, 180(3):219--236, 2006.

\end{thebibliography}
\bibliographystyle{alpha}

\newcommand{\etalchar}[1]{$^{#1}$}

\end{document}